\documentclass{vgtc}                          

\graphicspath{{figures/}{pictures/}{images/}{./}} 

\usepackage{times}                     

\usepackage{tabu}                      
\usepackage{booktabs}                  
\usepackage{lipsum}                    
\usepackage{mwe}                       

\usepackage{mathptmx}                  

\usepackage{multirow} 
\usepackage{siunitx} 
\usepackage{xspace}
\usepackage{subcaption}%
\usepackage{amsmath,amssymb,amsfonts}%
\usepackage{amsthm}%
\usepackage{algorithm}
\usepackage{algpseudocode}
\usepackage{comment}
\usepackage{xurl} 
\usepackage{tikz}
\usetikzlibrary{arrows.meta, positioning}

\usepackage{array} 

\newcommand{\real}{\mathbb{R}}

\newcommand{\pl}{PL\xspace}
\newcommand{\ep}{\varepsilon}

\newcommand{\rv}{\mathbf{r}}	
\newcommand{\sv}{\mathbf{s}}	
\newcommand{\qv}{\mathbf{q}}	
\newcommand{\spv}{\mathbf{s}'}	
\newcommand{\qpv}{\mathbf{q}'}	
\newcommand{\sd}{\mathsf{d}}
\newcommand{\rd}{\mathsf{m}}

\newcommand{\revision}[1]{#1}

\newtheorem{theorem}{Theorem}
\newtheorem{lemma}[theorem]{Lemma}

\onlineid{1010}

\vgtccategory{Research}

\vgtcinsertpkg

\title{Robust Geometric Predicates for Bivariate Computational Topology}

\author{
Petar Hristov\thanks{e-mail: petar.hristov@liu.se} 
\and Ingrid Hotz\thanks{e-mail: ingrid.hotz@liu.se}
\and Talha Bin Masood\thanks{e-mail: talha.bin.masood@liu.se}
}
\affiliation{\scriptsize Scientific Visualization Group, Department of Science and Technology (ITN), Link\"oping University, Norrk\"oping, Sweden}

\abstract{
    We present theory and practice for robust implementations of bivariate Jacobi set and Reeb space algorithms. 
    Robustness is a fundamental topic in computational geometry that deals with the issues of numerical errors and degenerate cases in algorithm implementations.
    Computational topology already uses some robustness techniques for the development of scalar field algorithms, such as those for computing critical points, merge trees, contour trees, Reeb graphs, Morse-Smale complexes, and persistent homology.
    In most cases, robustness can be ensured with floating-point arithmetic, and degenerate cases can be resolved with a standard symbolic perturbation technique called Simulation of Simplicity.
    However, this becomes much more complex for topological data structures of multifields, such as Jacobi sets and Reeb spaces.
    The geometric predicates used in their computation require exact arithmetic and a more involved treatment of degenerate cases to ensure correctness.
    Neither of these challenges has been fully addressed in the literature so far.
    In this paper, we describe how exact arithmetic and symbolic perturbation schemes can be used to enable robust implementations of bivariate Jacobi set and Reeb space algorithms.
    In the process, we develop a method for automatically evaluating predicates that can be expressed as large symbolic polynomials, which are difficult to factor appropriately by hand, as is typically done in the computational geometry literature.
    We provide implementations of all proposed approaches and evaluate their efficiency.
} 


\begin{document}


\firstsection{Introduction}

\maketitle

\revision{
Topological data structures like Reeb graphs~\cite{biasottiReebGraphsShape2008} and contour trees~\cite{CarrSnoeyinkAxen2003} are indispensable tools for identifying features in the analysis and visualization of scalar field data.
When more than one scalar field is defined over a common domain, which is often the case in practice, these data structures need to be generalized.
The set of critical points, where topological change occurs, is generalized to the Jacobi set~\cite{edelsbrunnerJacobiSetsMultiple2002a}.
Regions in data with uniform topological connectivity are captured by a generalization of the Reeb graph and contour tree, called the Reeb space~\cite{edelsbrunnerReebSpacesPiecewise2008a}.
For bivariate data, the Reeb space is a collection of two-dimensional sheets, attached to each other via edges and vertices (see \cref{fig:reeb-space-examples}).
}

\begin{figure}[hb]
    \centering    \includegraphics[width=.9\columnwidth]{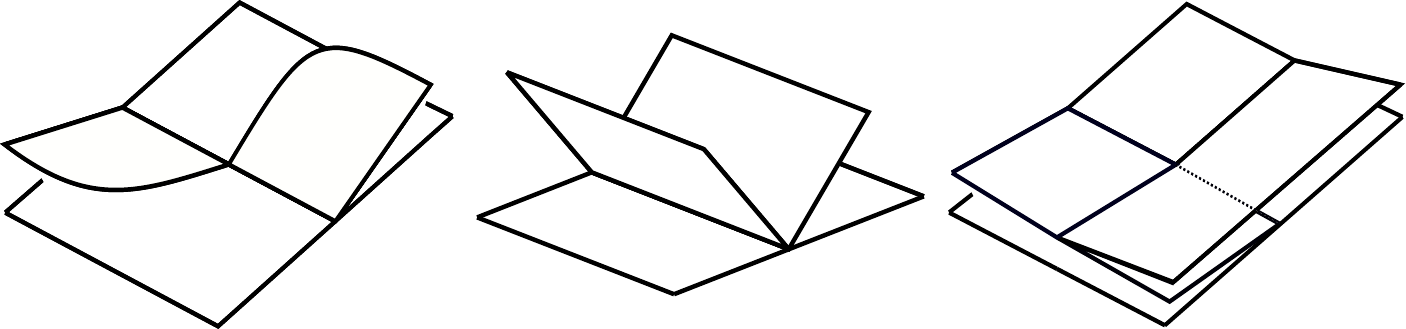}
    \caption{Examples of different Reeb spaces.}
    \label{fig:reeb-space-examples}
\end{figure}

\revision{
One of the main practical challenges in computing Jacobi sets and Reeb spaces is that all current algorithms~\cite{tiernyJacobiFiberSurfaces2017, chattopadhyayAlgorithmCorrectComputation2024, edelsbrunnerReebSpacesPiecewise2008a, hristovArrangeTraverseAlgorithm2025} assume that the input data is generic (i.e., mathematically well behaved).
For bivariate piecewise-linear~(\pl) data, typically this means that no two vertices map to the same point in the range, the images of no two edges are collinear, and the images of no three edges are concurrent (see \cref{fig:segment-degenerate}).
}
While the assumption of genericity greatly simplifies algorithm design, it is not realistic because degenerate cases do occur in real-world data and become more frequent as the size of the data increases.
If not handled correctly, even a small number of degenerate cases can change the combinatorial properties of a topological data structure, resulting in incorrect analysis and feature identification.

\begin{figure}[h]
    \centering
    \begin{subfigure}[b]{0.3\columnwidth}
        \includegraphics[width=\columnwidth]{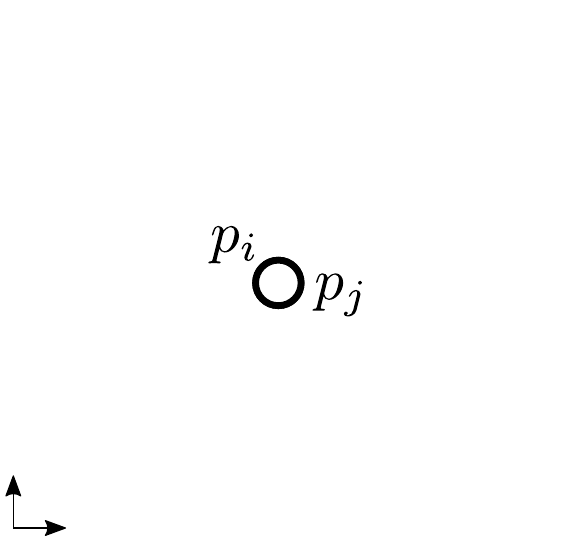}
        \caption{Duplicate points.}
        \label{fig:segment-order-1}
    \end{subfigure}
    \begin{subfigure}[b]{0.3\columnwidth}
        \includegraphics[width=\columnwidth]{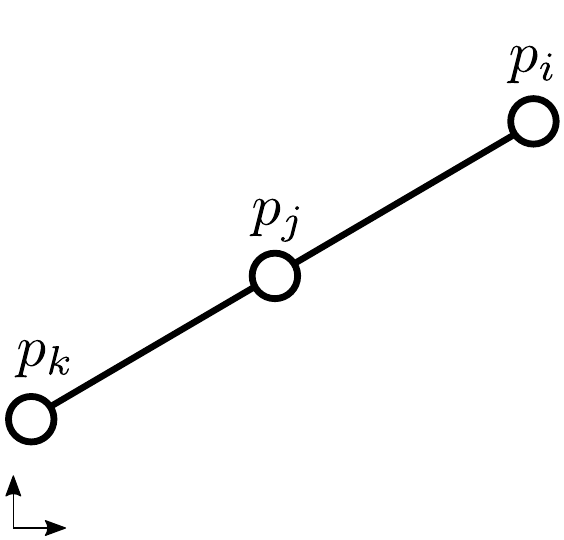}
        \caption{Collinear segments.}
        \label{fig:sub2}
    \end{subfigure}
    \begin{subfigure}[b]{0.3\columnwidth}
        \includegraphics[width=\columnwidth]{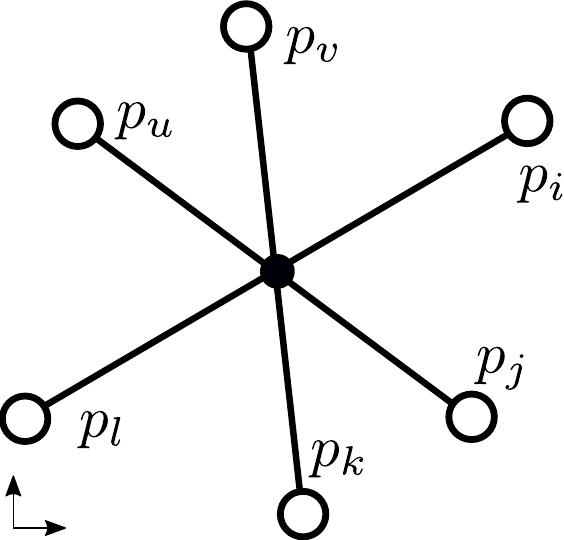}
        \caption{Concurrent segments.}
        \label{fig:degenerte-cases}
    \end{subfigure}
    \caption{
        Examples of the degenerate cases for generic bivariate \pl maps. 
        Each point and segment in the plane is the image of a vertex and edge of the input mesh.
    }
    \label{fig:segment-degenerate}
\end{figure}





In practice, genericity is often enforced by perturbing degenerate data values into general position.
The standard perturbation scheme used for this purpose is called \emph{Simulation of Simplicity}~(SoS)~\cite{edelsbrunnerSimulationSimplicityTechnique1990a}, which symbolically adds infinitesimal numbers $\ep_i$ to the input values.
These $\ep$-terms are never explicitly evaluated; they are only used to break ties when the original values are not distinct.
Since data is sampled at finite resolution, adding infinitesimal values with SoS preserves all non-degenerate features, which are usually the ones of interest in applications.
Spurious features introduced by resolving degenerate configurations can be removed with simplification methods~\cite{edelsbrunnerTopologicalPersistenceSimplification2002, lukasczyk2021}.

In the scalar case, SoS can be implemented elegantly by breaking ties between equal data values using their unique memory indices.
This is enough to resolve all degenerate cases in most scalar field computational topology algorithms.
Since floating-point comparisons are sufficiently robust in practice, this pre-processing step is seamlessly integrated into most computational topology libraries~\cite{tierny2018, moreland2016}.
In fact, SoS is so ubiquitous that it is often cited~\cite{tiernyJacobiFiberSurfaces2017, chattopadhyayAlgorithmCorrectComputation2024, sharmaJacobiSetDriven2022c, hristovArrangeTraverseAlgorithm2025} as the standard method to ensure genericity in multivariate topological algorithms, such as those for Jacobi sets and Reeb spaces, without providing implementation details. 
However, extending SoS to multivariate data is nontrivial.

\begin{figure*}[h!]
\centering
\scalebox{0.8}{
\begin{tikzpicture}[
  node distance=0.3cm and 0.7cm,
  every node/.style={draw, rounded corners, text width=3.6cm, align=center, minimum height=1cm},
  arrow/.style={-{Latex}}
]

\definecolor{lightblue}{RGB}{173,216,230}
\definecolor{lightgreen}{RGB}{180,238,180}
\definecolor{lightorange}{RGB}{255,204,153}
\definecolor{lightgray}{RGB}{211,211,211}
\definecolor{lightpurple}{RGB}{221,160,221}
\definecolor{lightred}{RGB}{255,182,193}

\node (topology) [draw=none, minimum height=0.5cm] {\bf Bivariate Topology};
\node (Jacobi) [below=of topology] {Jacobi set\\ \cref{sec:multifield-topo}};
\node (Reeb) [below=of Jacobi] {Reeb space\\ \cref{sec:multifield-topo}};

\node (pred) [draw=none, right=of topology, minimum height=0.5cm] {\bf Predicates};
\node (ulink) [below=of pred] {Upper link \\ \cref{sec:arrangements}};
\node (segment) [below=of ulink] {Segment order \\ \cref{sec:arrangements}};

\node (algebra) [draw=none, right=of pred, minimum height=0.5cm] {\bf Algebraic formulations};
\node (orient) [below=of algebra] {Point orientation\\ \cref{sec:arrangements}, \cref{eqn:point_orientation}};
\node (parameter) [below=of orient] {Parametrization\\
\cref{sec:arrangements}, \cref{fig:segment-order}};
\node (dual) [below=of parameter] {Dualization\\
\cref{sec:predicates-reeb-space}, \cref{eqn:segment_order}};

\node (dummy1) [draw=none, right=of algebra,text width=0.1cm,] { };

\node (scheme) [draw=none, right=of dummy1, minimum height=0.5cm] {\bf Perturbation schemes};
\node (sos) [below=of scheme] {Simulation of simplicity\\ \cref{sec:symbolic-perturb}, \cref{sec:sos}};
\node (general) [below=of sos] {Perturbing the world \\\cref{sec:symbolic-perturb}, \cref{sec:allienz}};
\node (perturb) [below=of general] {General symbolic scheme\\ \cref{sec:symbolic-perturb}, \cref{sec:yap}};

\draw[arrow] (Jacobi) -- (ulink);
\draw[arrow] (Reeb) -- (segment);
\draw[arrow] (Reeb.east) -- (ulink.west);
\draw[arrow] (ulink) -- (orient);
\draw[arrow] (segment.east) -- (parameter.west);
\draw[arrow] (segment.east) -- ++(0.2cm,0) |- (dual.west);


\draw[arrow] (orient.east) -- (sos.west);
\draw[arrow] (orient.east) -- (general.west);
\draw[arrow] (orient.east) -- (perturb.west);
\draw[arrow] (parameter.east) -- (sos.west);
\draw[arrow] (parameter.east) -- (general.west);
\draw[arrow] (parameter.east) -- (perturb.west);
\draw[arrow] (dual.east) -- (sos.west);
\draw[arrow] (dual.east) -- (general.west);
\draw[arrow] (dual.east) -- (perturb.west);

\end{tikzpicture}
}
\caption{Overview: This work primarily investigates two bivariate topological structures -- Jacobi sets and Reeb spaces. The computation of these structures relies on two key predicates, for which we present various algebraic formulations. Finally, we examine three distinct perturbation schemes aimed at ensuring robust implementation of these predicates.}
\label{fig:overview}
\end{figure*}
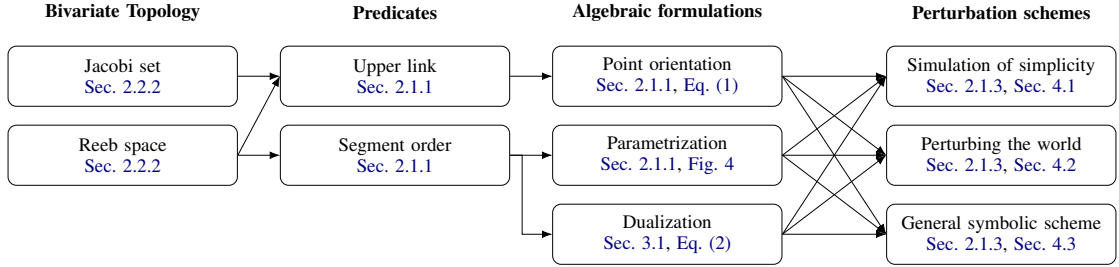

%
%
Symbolic perturbation techniques, including SoS, face several challenges even in two or three dimensions. First, even simple predicates like orientation tests can be numerically unstable with floating-point arithmetic, necessitating exact arithmetic -- typically using arbitrary-precision rationals
\cite{goodmanHandbookDiscreteComputational2018}. 
Second, predicates involving derived values (e.g., segment intersections) yield complex polynomials in the symbolic $\ep$-variables, often with hundreds of terms. These are notoriously difficult to handle manually, limiting the practicality of many symbolic perturbation schemes~\cite{yapSymbolicTreatmentGeometric1990, emirisEfficientPerturbationsHandling1997}.

In this work, we address the challenge of ensuring genericity in multivariate topological algorithms. We develop both the theoretical framework and practical tools required for the robust computation of geometric predicates used in bivariate Jacobi set and Reeb space algorithms. Our approach is compatible with two foundational symbolic perturbation techniques: SoS~\cite{edelsbrunnerSimulationSimplicityTechnique1990a} and general symbolic perturbation~\cite{yapGeometricConsistencyTheorem1990, yapSymbolicTreatmentGeometric1990}. This enables us to transform any degenerate bivariate \pl map $f$ into a generic map $f_{\ep}$, thereby ensuring the robustness of the existing Jacobi set and Reeb space algorithms against both numerical inaccuracies and degeneracies.
%
%
For predicates with large symbolic expressions that are impractical to resolve manually, we introduce automated evaluation via symbolic math tools such as Mathematica or SymPy. These expressions are organized into evaluation tables~\cite{edelsbrunnerSimulationSimplicityTechnique1990a}, from which efficient C++ code can be generated and integrated into existing systems. All predicates are implemented and evaluated.

We review background from computational geometry and topology in \cref{sec:background}. The predicates for computing Jacobi sets and Reeb spaces are outlined in \cref{sec:required-predicates}, and their perturbation via symbolic schemes is described in \cref{sec:schemes}. Genericity of perturbed maps is proved in \cref{sec:genericity}. \Cref{sec:polynomial-sign} introduces our approach to evaluating large symbolic expressions. Implementation and performance are covered in \cref{sec:implementation,sec:evaluation}, followed by discussion and conclusions in \cref{sec:discussion,sec:conclusion}.
\cref{fig:overview} provides a different perspective on the organization of this paper, highlighting the key steps required for the robust computation of bivariate topological structures.


\section{Background}
\label{sec:background}




Before starting our discussion, we will introduce our main definitions.
A \emph{triangulation of a $\sd-$manifold} is a $\sd$-dimensional simplicial complex $M$  whose underlying geometric space is homeomorphic to a $\sd$-manifold \revision{(see \cref{app:background-pl})}.
A \pl function $f: |M| \to \real$ whose range is the one-dimensional Euclidean space is called \textit{scalar}.
A collection of scalar \pl functions defined over a common domain is called a \textit{multifield}~\cite{Fuchs2009}. Equivalently, a multifield can be viewed as a multivariate 
\pl map $f = (f_1, \dots,f_\rd): M \to \real^\rd$, where each $f_i$ is called a \textit{component function}.
In this paper, we will primarily focus on multifields with two components, which are known as \textit{bivariate fields}.

Let $V = \{v_1, \dots, v_{n_0}\}$ and $E = \{e_1, \dots, e_{n_1}\}$ be the vertices and edges of $M$ respectively.
We denote the images of $V$ and $E$ in $\real^\rd$ as $P = \{p_1, \dots, p_{n_0}\}$ and $S = \{s_1, \dots, s_{n_1}\}$, where $p_i = f(v_i)$ for $1 \leq i \leq n_0$ and $s_i = f(e_i)$ for $1 \leq i \leq n_1$.
We will refer to the elements of $P$ and also more generally of the range space as \textit{points}, and to the elements of $S$ as line segments, or simply \textit{segments}.
Each point has coordinates $p_i = (p_{i,1}, \dots, p_{i,\rd})$
We will refer to the edges of $M$ and the segments in $\real^\rd$ using their endpoints -- if an edge $e$ has endpoints $v_i$ and $v_j$, we write $v_iv_j = e \in M$, and $s_{i, j} = f(e) \subset \real^\rd$.








We proceed with background from computational geometry in \cref{sec:comp-geom} and computational topology in \cref{sec:computational-topology}.

\subsection{Computational Geometry}
\label{sec:comp-geom}
Computational geometry algorithms have two key operations: \emph{predicates} and \emph{constructions}.
A geometric predicate is a function that determines a combinatorial relationship between geometric objects when an algorithm has to make a decision.
These include, for example, orientation tests or intersection tests.
Predicates are usually evaluated on the basis of the sign of a determinant or, more generally, the sign of a polynomial.
Constructions are made on the basis of the outcomes of predicates.
For example, adding new vertices at intersection points or new lines on the boundaries of convex hulls.

\subsubsection{Arrangements of Line Segments}
\label{sec:arrangements}

A fundamental problem in computational geometry is partitioning the Euclidean plane into regions, called faces, using a set of lines or line segments.
The resulting partition is known as an \textit{arrangement} and is typically computed in three stages: (1) finding segment intersections, often using a plane sweep algorithm~\cite{goodmanHandbookDiscreteComputational2018}; (2) determining the relative order of these intersections; and (3) building a doubly connected edge list (DCEL) to represent the arrangement~\cite{debergComputationalGeometryAlgorithms2008a}.
Computing an arrangement of line segments can be done in $O(n\log n + k\log n)$ time~\cite{goodmanHandbookDiscreteComputational2018}.
Arrangements can be computed for many different types of geometric objects such as triangles, hyperplanes, circles, and more general curves~\cite{goodmanHandbookDiscreteComputational2018}.



\paragraph*{Required Predicates.}
Arrangement algorithms for line segments rely on three fundamental predicates: \textit{lexicographic order of points}, \textit{orientation of points} and \textit{order of segments}.
Lexicographic order compares two points by their first component and, if those are equal, then by their second component.
It is used by plane sweep algorithms to maintain an event queue of intersection points and segment endpoints.
Orientation tests determine which half-plane a point $p_k = (p_{k, 1}, p_{k, 2})$ is in with respect to the line defined by two other points $p_i = (p_{i, 1}, p_{i, 2})$ and $p_j = (p_{j, 1}, p_{j, 2})$.
This can be evaluated by computing the determinant of the following matrix:
\begin{equation}
    \triangle_p =
    \begin{pmatrix}
        p_{i, 1} & p_{i, 2} & 1 \\
        p_{j, 1} & p_{j, 2} & 1 \\
        p_{k, 1} & p_{k, 2} & 1 \\   
    \end{pmatrix}
    .
    \label{eqn:point_orientation}
\end{equation}
Orientation tests are used to detect intersections between segments and to order segments around a vertex in the DCEL data structure.
The third predicate, segment order, determines the relative order in which the three segments intersect each other.
Although it can be derived from orientation tests, doing so requires computing intersection points, which is costly. A more efficient alternative is to parametrize the segments and compare the parameter values (see \cref{fig:segment-order}).
Rearranging the resulting expression to eliminate the denominators avoids division, improving both performance and robustness.


\begin{figure}[h]
    \centering
    \begin{subfigure}[b]{0.47\columnwidth}
        \includegraphics[width=\columnwidth]{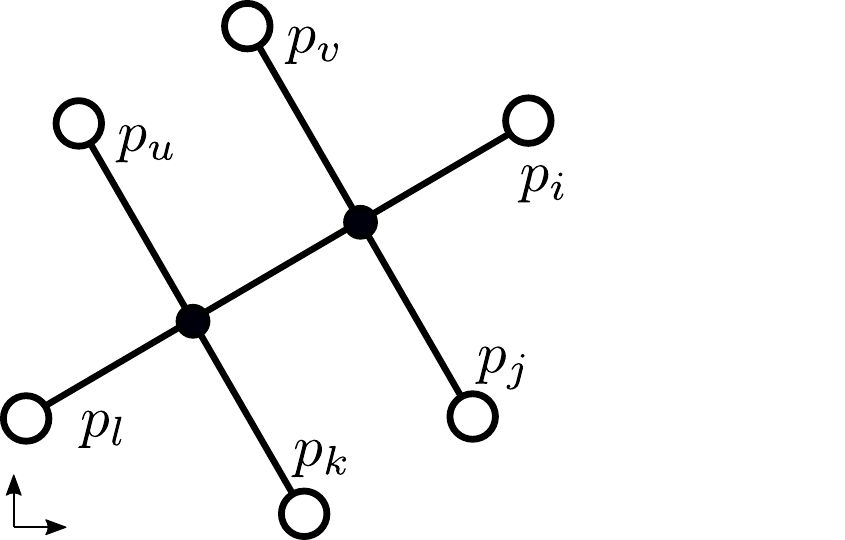}
        \caption{Intersecting segments.}
        \label{fig:segment-order-1}
    \end{subfigure}
    \begin{subfigure}[b]{0.47\columnwidth}
        \includegraphics[width=\columnwidth]{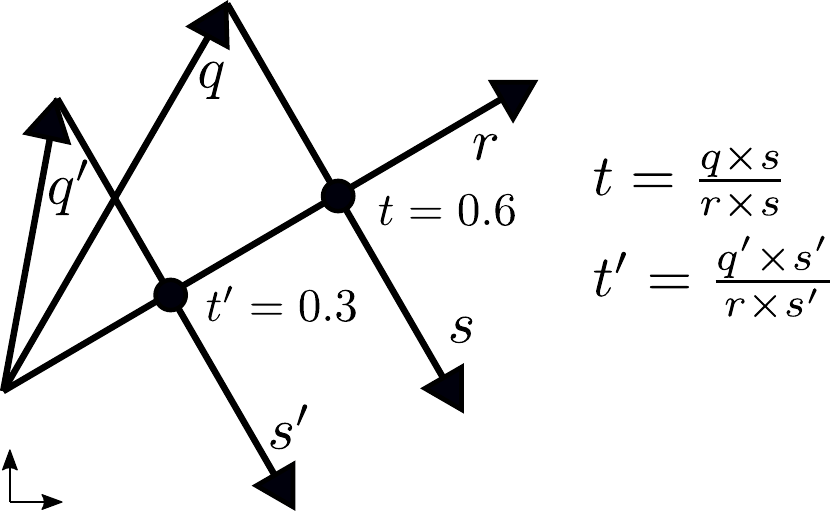}
        \caption{Segment vectors.}
        \label{fig:sub2}
    \end{subfigure}
    \caption{
        The relative order of intersection of three segments can be determined by comparing the parameters $t$ and $t'$ of the intersection points over a parametrization of the first segment.
        We compute $t$ and $t'$ using the scalar cross product $\times$ with $r, s, s', r$ and $r'$.
    }
    \label{fig:segment-order}
\end{figure}

\subsubsection{Robustness}
Computational geometry algorithms typically assume infinite numerical precision and inputs in general position.
These assumptions simplify exposition, but are impractical in real implementations. Ensuring correctness despite numerical errors or degeneracies is known as \emph{robustness}. Robustness is critical for predicates that determine the combinatorial structure of the output.

\paragraph*{Numerical Stability.}
It is well known that floating point arithmetic suffers from precision errors, particularly when comparing expressions that are close to zero, as is the case with geometric predicates.
We encourage the reader to consult the following examples of how orientation tests can fail due to such inaccuracies~\cite[Section 3]{kettnerClassroomExamplesRobustness2008}.
When predicates fail, they can compromise the accuracy of algorithms, potentially leading to errors in the global combinatorial structure of the computed solution.

Robustness in the presence of numerical errors is a broad topic. Here, we briefly highlight some key concepts and refer the reader to a comprehensive overview in~\cite{goodmanHandbookDiscreteComputational2018} for further details.
A central approach in this area is the \emph{exact geometric computation (EGC)} paradigm, which is central to many modern implementations of computational geometry algorithms, such as those found in the open-source CGAL library~\cite{cgal-manual}.
EGC avoids numerical errors arising from floating-point arithmetic by using exact arithmetic with rational number types and arbitrary-precision integers.
To avoid the performance bottleneck of exact arithmetic, typically \emph{filters} are employed to first compute predicates approximately with floating-point numbers.
Exact arithmetic is then only used when the result of the floating-point computation is outside a controlled error bound.

\paragraph*{Degenerate input configurations.}
A remaining challenge for robustness is handling degenerate input configurations.
Degeneracies are algorithm-dependent. 
Some algorithms may fail on collinear points, while others may not. 
The most direct and often preferred approach in practical implementations is to explicitly handle degenerate cases as usual parts of the input~\cite{cgal-manual}.
However, this is not always feasible, especially when a complete classification of all possible degeneracies and their effects on the algorithm's output is unavailable.
As an alternative, there are various techniques that slightly perturb the input into general position.
The simplest form of input perturbation is \emph{random numeric perturbation}, where small random values are added to the input data.
Although this approach can resolve many degenerate cases, it is not guaranteed to eliminate all of them due to finite precision and its random nature.
In contrast, \emph{controlled perturbation} uses adaptive precision to exactly compute a non-degenerate input configuration~\cite{mehlhornGeneralApproachAnalysis2011}.
Finally, several \emph{symbolic perturbation} schemes exist in which the actual perturbation is never explicitly calculated.
Instead, the output of a predicate is decided as if the input was perturbed by an infinitesimally small amount.
The goal is to ensure that the perturbed objects are in general position, to keep all non-degenerate properties of the original data, and to lend to an efficient implementation~\cite{edelsbrunnerSimulationSimplicityTechnique1990a}.

\subsubsection{Symbolic Perturbation}
\label{sec:symbolic-perturb}

One of the first symbolic perturbation schemes is called Simulation of Simplicity (SoS)~\cite{edelsbrunnerSimulationSimplicityTechnique1990a}.
Given an input with a number of geometric objects, each specified by a number of \emph{coordinates} $p_i = (p_{i, 1}, \dots, p_{i, \rd})$, the \emph{$\ep$-expansion} of $p_i$ is defined as $p_i(\ep) = (p_{i, 1} + \ep_{i, 1}, \dots, p_{i, \rd} + \ep_{i, \rd})$, where each $\ep_{i, j}$ is an infinitesimal number called an \emph{$\ep$-variable}, for $j \in \{1, \dots, \rd\}$.
The product of multiple $\ep$-variables $\ep_{i_1, j_1}\ep_{i_2, j_2}\cdots\ep_{i_n, j_n}$ is called an \emph{$\ep$-product} if all $\ep_{i_k, j_k}$ are unique for $k \in \{1, \dots, n\}$.
The specific choice of $\ep$-variables is not important, as long as they satisfy properties that establish a lexicographic order among $\ep$-variables and their products, which reflects differences in their orders of magnitude.

In the SoS framework, many predicates expressed as determinants are guaranteed to be nonzero. 
For example, for an orientation test, this ensures that no three points are treated as collinear.
To determine the sign of a perturbed predicate, SoS successively evaluates \emph{subdeterminants} from the original determinant.
These subdeterminants are lexicographically ordered according to their associated $\ep$-products and organized into a \emph{table of relevant determinants}.
Once a \revision{constant} non-zero subdeterminant is found, the evaluation stops, as smaller $\ep$-products cannot affect the overall sign.
The number of rows in the table is referred to as the maximum \emph{depth}.

A more general approach that allows predicates to be expressed as rational polynomials is given by Yap's \emph{general symbolic scheme}~\cite{yapGeometricConsistencyTheorem1990, yapSymbolicTreatmentGeometric1990}, which avoids explicitly specifying an $\ep$-perturbation. 
It gives a general method for evaluating the perturbed sign $s \in \{-1, 1\}$ of a rational polynomial $f(x_1, \dots, x_n)$ with input $q \in \mathbb{Q}^n$ even when $f(q) = 0$.
This is achieved by establishing a total \emph{admissible} ordering on the partial derivatives of $f$, and then differentiating $f$ with respect to the ordered partial derivatives until a nonzero value is found.
Two examples of total admissible orderings are \emph{lexicographic ordering} and \emph{total ordering}, where the total order is first compared and then the lexicographic ordering is used to break ties.
The partial derivatives can be computed symbolically with computer algebra packages~\cite{seidelNatureMeaningPerturbations1998}.

\revision{
For example, \cref{table:scheme-comparison} shows a direct comparison between SoS and the general symbolic scheme for an orientation test of two points on the projective line in homogeneous coordinates.
Each row in the table for SoS lists subdeterminants that multiply $\ep$-products in decreasing order of significance~\cite[Eq. 6, Sec. 4]{edelsbrunnerSimulationSimplicityTechnique1990a}.
For the general scheme, we list the expressions that correspond to totally ordered partial derivatives.
}

\begin{table}[ht]
\centering
\scalebox{0.8}{
\begin{tabular}{lllll}
\toprule
& \multicolumn{2}{c}{SoS~\cite{edelsbrunnerSimulationSimplicityTechnique1990a}} & \multicolumn{2}{c}{General scheme~\cite{yapSymbolicTreatmentGeometric1990}} \\
\cmidrule(lr){2-3} \cmidrule(lr){4-5}
Row & Subdeterminant & $\ep$-product & Expression & Derivative \\
\midrule
0 & $\begin{vmatrix} p_{i, 1} & p_{i, 2} \\ p_{j, 1} & p_{j, 2} \end{vmatrix}$ & $\emptyset$ & $\begin{vmatrix} p_{i, 1} & p_{i, 2} \\ p_{j, 1} & p_{j, 2} \end{vmatrix}$ & $\emptyset$ \\
1 & $p_{j, 1}$ & $\ep_{i, 2}$ & $p_{j, 2}$ & $\partial p_{i, 1}$ \\
2 & $-p_{j, 2}$ & $\ep_{i, 1}$ & $-p_{j, 1}$ & $\partial p_{i, 2}$ \\
3 & $p_{i, 1}$ & $\ep_{i, 2}$ & $-p_{i, 2}$ & $\partial p_{j, 1}$ \\
4 & $1$ & $\ep_{j, 2}\ep_{i, 1}$ & $p_{i, 1}$ & $\partial p_{j, 2}$ \\
5 &  &  & 1 & $\partial p_{i, 1}p_{j, 2}$ \\
\bottomrule
\end{tabular}
}
\caption{
    Evaluation table comparisons between SoS and Yap's general symbolic scheme for the scalar cross product of two vectors.
}
\label{table:scheme-comparison}
\end{table}

The subsequent work~\cite{emirisEfficientPerturbationsHandling1997} improves efficiency by reducing the algebraic degree of perturbed expressions using a \emph{linear perturbation} $p_{i, j}(\ep) = p_{i, j} + i^j \ep$.
This scheme is less general than SoS but has been successfully applied to orientation tests and in-sphere tests. 
More recent work~\cite{irvingDeterministicPseudorandomPerturbation2013} has focused on combining the efficiency of linear perturbations~\cite{emirisEfficientPerturbationsHandling1997} with the general symbolic scheme~\cite{yapSymbolicTreatmentGeometric1990} or on geometric symbolic perturbation schemes with qualitative symbolic perturbation~\cite{devillersQualitativeSymbolicPerturbation2016}.
Finally, we note that all symbolic perturbation schemes assume exact computation~\cite{devillersQualitativeSymbolicPerturbation2016, seidelNatureMeaningPerturbations1998}.

\paragraph*{Symbolic Perturbation of Lines and Line Segments.}
Symbolic perturbation can be extended to line segments by representing each segment as a point in $\real^4$, formed by concatenating its endpoints.
This resolves common degeneracies such as segments with endpoints or intersection points that share the $x$- or $y$-coordinates, collinear segments, concurrent segments, and segments that share endpoints.
Alternatively, perturbing only the endpoints of segments addresses the first three types of degeneracies, while preserving the incidence relationships between segments that share endpoints~\cite{alliezRemovingDegeneraciesPerturbing2000}.
While the paper does not provide an explicit perturbation scheme for segments, it suggests that it would be similar to the one used for 2D Delaunay triangulations:
$ p_i(\ep) = (p_{i, 1} + \ep p_{i, 2}, p_{i, 2} + \ep^2p_{i, 1} + \ep^3(p_{i, 1}^2 + p_{i, 2}^2)).$
The technique is referred to as \emph{perturbing the world} and applies a single global $\ep$-variable. \cref{fig:perturbation-vertex-segment} shows the two alternatives for symbolic perturbation of line segments.

\begin{figure}[h]
    \centering
    \begin{subfigure}[b]{0.31\columnwidth}
        \includegraphics[width=\columnwidth]{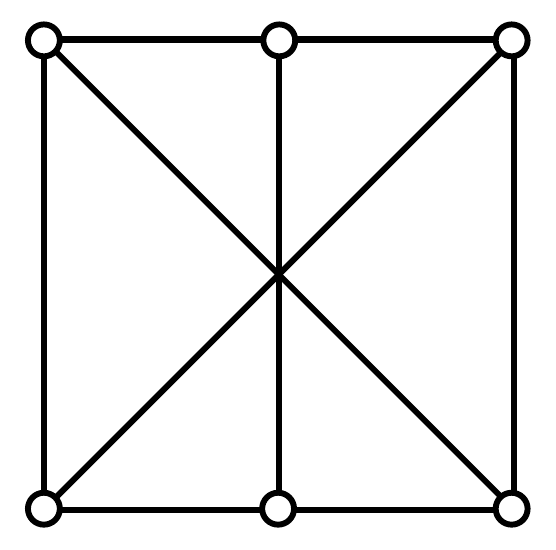}
        \caption{Input segments.}
        \label{fig:segment-order-1}
    \end{subfigure}
    \begin{subfigure}[b]{0.31\columnwidth}
        \includegraphics[width=\columnwidth]{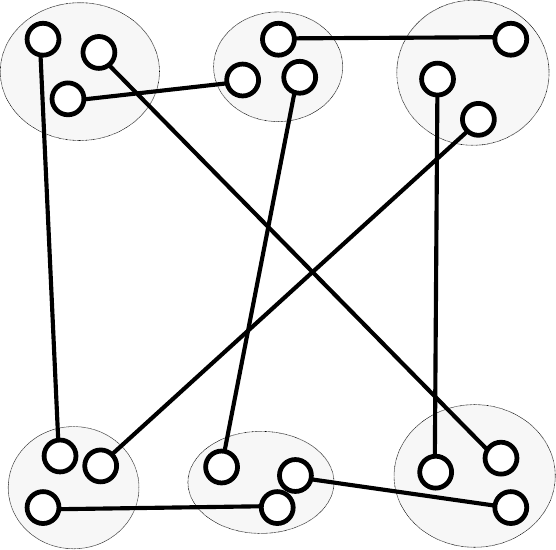}
        \caption{Segment perturbation.}
        \label{fig:sub2}
    \end{subfigure}
    \begin{subfigure}[b]{0.31\columnwidth}
        \includegraphics[width=\columnwidth]{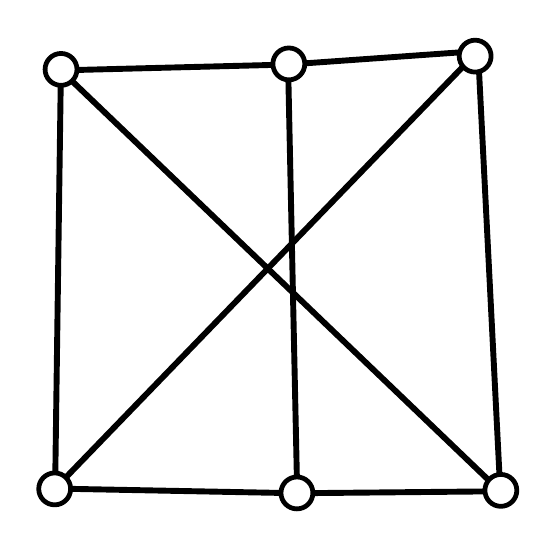}
        \caption{Endpoint perturbation.}
        \label{fig:sub2}
    \end{subfigure}
    \caption{
        Two perturbations of line segments that remove different types of degenerate cases.
        Perturbing the segments (b) as four-dimensional points resolves all degenerate cases, but segments that used to share endpoints no longer do.
        Perturbing the endpoints of segments (c) keeps the overlapping endpoint incidence relationship between segments, while removing all other degenerate cases.
    }
    \label{fig:perturbation-vertex-segment}
\end{figure}







\subsection{Computational Topology}
\label{sec:computational-topology}

In this section, we review key results from computational topology for both scalar fields and multi-fields, with an emphasis on Jacobi set and Reeb space algorithms.
We give special attention to generic piecewise-linear (\pl) maps and their various definitions in the literature.
For a general introduction to \pl and computational topology, we refer the reader to established textbooks~\cite{edelsbrunnerComputationalTopology2009, rourkeIntroductionPiecewiseLinearTopology1982a}.


\subsubsection{Scalar Field Topology}
Scalar field topology studies the structure of real-valued functions $f: |M| \to \real$ called scalar fields, where the input data $M$ is typically represented by a regular grid or a triangulated manifold of dimension $\sd$ \revision{(see \cref{app:background-pl})}.
The main topological structures used to extract and analyze features in scalar fields are merge trees~\cite{BeketayevYeliussizovMorozov2014}, contour trees~\cite{CarrSnoeyinkAxen2003}, Reeb graphs~\cite{biasottiReebGraphsShape2008}, Morse-Smale complexes~\cite{EdelsbrunnerHarerNatarajan2003}, and persistent homology~\cite{Edelsbrunner2008}.
The success of these methods is largely due to efficient combinatorial algorithms for their computation, which require a robust handling of geometric predicates.

The core geometric predicate used by these algorithms is a comparison that determines the upper and lower link of a vertex.
The upper and lower link partition the vertices of the link into those with higher function value (upper link) and lower function value (lower link).
These predicates can be evaluated efficiently and robustly using floating-point arithmetic, requiring only a single comparison per link vertex. To handle degenerate cases, such as vertices with identical function values, SoS is used to break the ties~\cite{edelsbrunnerComputationalTopology2009}.

Simulation of Simplicity~\cite{edelsbrunnerSimulationSimplicityTechnique1990a} simplifies algorithm design by ensuring a scalar field is \emph{generic}, that is, free of degeneracies. 
Symbolic perturbation, however, comes at a cost: it can introduce artificial features, particularly in flat regions.
To remove these, a subsequent \emph{simplification} step is applied, typically guided by a geometric measure such as persistence, volume, or hypervolume.

\subsubsection{Multifield Topology}
\label{sec:multifield-topo}



In many real-world applications, the key features of interest arise from complex interactions among the component scalar fields in multifield data. Since these relationships cannot be captured by the topological descriptors of individual fields or their pairwise combinations, growing attention has been devoted to extending topological data structures to the multivariate setting.
%
Two topological data structures have been successfully generalized with efficient algorithms:
Jacobi sets~\cite{sharmaJacobiSetDriven2022c, tiernyJacobiFiberSurfaces2017}, which extend the concept of critical points, and Reeb spaces~\cite{hristovArrangeTraverseAlgorithm2025}, which generalize Reeb graphs and contour trees.
However, research on the robustness and reliability of these algorithms remains in its early stages.



\paragraph*{Jacobi Sets.}
The Jacobi set of a multifield $f = (f_1, \dots, f_\rd): |M| \to \real^\rd$, where $M$ is a triangulation of a $\sd$-manifold, generalizes the notion of critical vertices in the scalar case to Jacobi simplices~\cite{edelsbrunnerJacobiSetsMultiple2002a}. 
A simplex $\sigma \in M$ of dimension $\sd-\rd$ is part of the Jacobi set, depending on the topology of its upper and lower link, computed in terms of reduced Betti numbers.

The geometric predicate  used to decide whether a vertex $v \in \textnormal{Link}(\sigma)$ is in the upper or lower link of $\sigma$ uses a scalar function $h_{\lambda}: \textnormal{Link}(\sigma) \to \real$.
The vertex $v$ belongs to the upper link if $h_{\lambda}(v) > h_{\lambda}(\sigma)$, and to the lower link if $h_{\lambda}(v) < h_{\lambda}(\sigma)$.
In the bivariate case for an edge $ab \in M$, the scalar function is defined as $h_{\lambda}(v) = \lambda f_1(v) + f_2(v)$, where $\lambda$ is the negative slope of the line between the points $f(a)$ and $f(b)$, as shown in \cref{fig:jacobi-set-orientation-example}.
%
In practice, this predicate can be computed with an orientation test~\cite{hristovHypersweepsConvectiveClouds2022} using the points $f(a), f(b)$ and $f(v)$ or
using the range-based signed distance field with respect to the image of the line segment $f(ab)$ in the range~\cite{sharmaJacobiSetDriven2022c, tiernyJacobiFiberSurfaces2017}.

\begin{figure}[h]
    \includegraphics[width=.9\columnwidth]{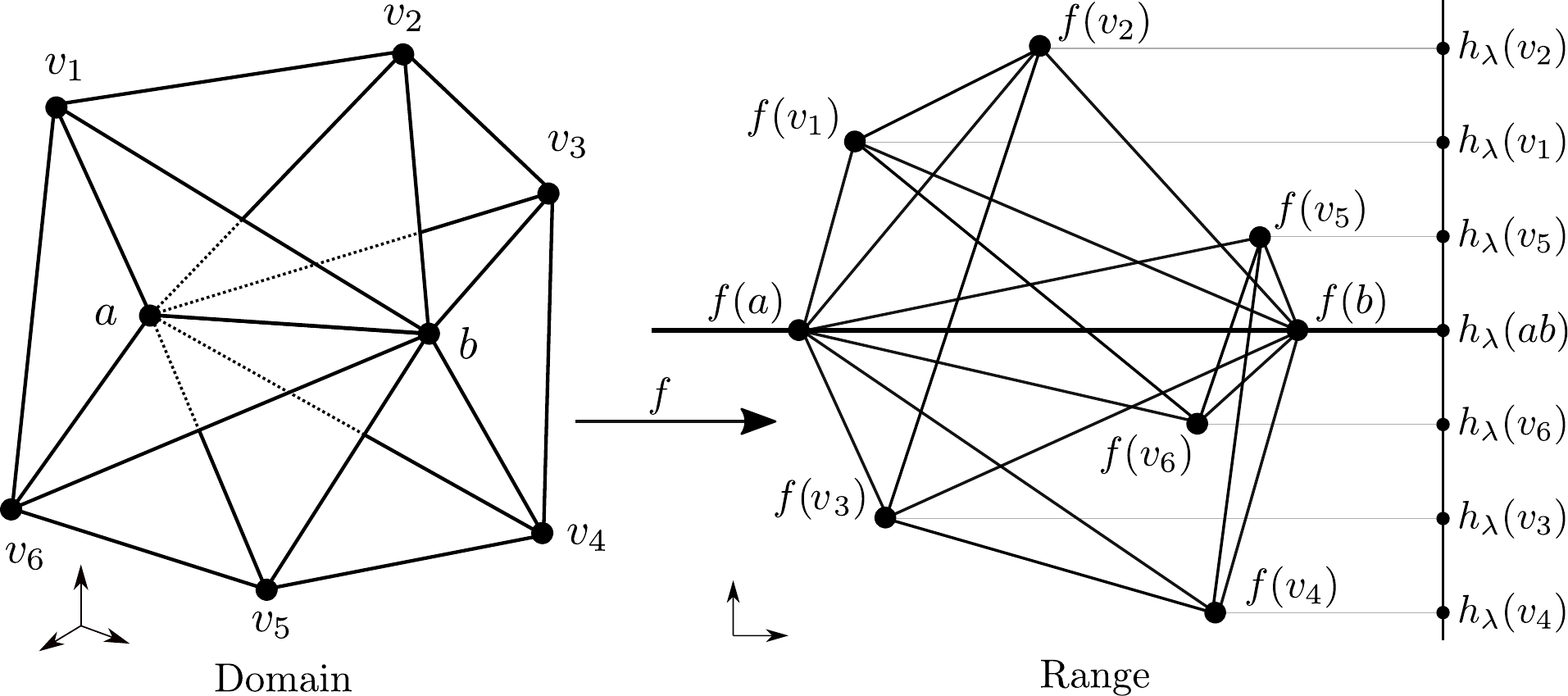}
    \caption{
    Geometric interpretation of the scalar function $h_{\lambda}$ used to define Jacobi edges.
    Vertices mapped above the line defined by $f(ab)$ belong to the upper link of $ab$, and the ones that map below belong to the lower link. 
    }
    \label{fig:jacobi-set-orientation-example}
\end{figure}


Unlike in the scalar case, robustly handling degenerate situations where $h_{\lambda}(v) = h_{\lambda}(ab)$ is considerably more challenging.
To the best of our knowledge, this issue has only been addressed for Jacobi sets of bivariate maps defined over triangulations of 2-manifolds~\cite{natarajanTopologicalAnalysisScalar2004}.
Given a bivariate \pl map $(f_1,f_2): |M| \to \real^2$, the authors use a symbolic perturbation scheme similar to SoS. 
The $\ep$-expansion is given as $f_1(v_i) + \ep_i$ and $f_2(v_i) + \delta_i$, where $v_i$ is the vertex of the triangulation with index $i$.
The order of the symbolic variables is $\ep_1 \gg \dots \gg \ep_n \gg \delta_1 \gg \dots \gg \delta_n$, where $n$ is the number of vertices.
The definition of $h_{\lambda}(v_i)$ is then expanded in terms of the symbolically perturbed coordinates and the necessary if-else statements to robustly resolve the predicate are derived.


\paragraph*{Reeb spaces.}
The Reeb space generalizes the Reeb graph by encoding the connectivity of connected components of level sets of multivariate \pl maps.
These level sets are known as \textit{fibers}~\cite{saekiTopologySingularFibers2004b} and their connected components are called \textit{fiber components}.
In the bivariate case over a triangulation of a 3-manifold, the fibers are generically a set of one-dimensional \pl curves, except when they intersect Jacobi edges. 
Fiber components intersecting Jacobi edges are either isolated points (analogous to local maxima and minima in the scalar case) or collections of \pl curves all meeting at a single point (analogous to saddles in the scalar case).

For a multifield with $\rd$ components, the Reeb space is locally an $\rd$-manifold, except at Jacobi edges, where it becomes singular~\cite{edelsbrunnerReebSpacesPiecewise2008a}.
In the bivariate case, the Reeb space consists of two-dimensional sheets glued together at the Jacobi edges in possibly complicated ways; see~\cref{fig:reeb-space-examples}.
For a comprehensive introduction to Reeb spaces, we refer the reader to the following works~\cite{hristovHypersweepsConvectiveClouds2022, edelsbrunnerReebSpacesPiecewise2008a}.
Throughout this paper, we will only consider 2-dimensional Reeb spaces, although the concept and computation extend to higher dimensions~\cite{edelsbrunnerReebSpacesPiecewise2008a, hristovArrangeTraverseAlgorithm2025}.

There are four algorithms for computing the Reeb space~\cite{edelsbrunnerReebSpacesPiecewise2008a, tiernyJacobiFiberSurfaces2017, hristovArrangeTraverseAlgorithm2025, chattopadhyayAlgorithmCorrectComputation2024}.
Two of them~\cite{edelsbrunnerReebSpacesPiecewise2008a, hristovArrangeTraverseAlgorithm2025} construct an \emph{arrangement} in the range, then proceed in a fully combinatorial manner.
As a result, the geometric predicates they rely on are the same as those for computing the arrangement: lexicographic order, orientation of points, and order of segments.
The other two algorithms~\cite{chattopadhyayAlgorithmCorrectComputation2024, tiernyJacobiFiberSurfaces2017} avoid building the full arrangement, but still use similar geometric predicates.
The first~\cite{chattopadhyayAlgorithmCorrectComputation2024} computes segment intersections and their relative order.
The second~\cite{tiernyJacobiFiberSurfaces2017} also identifies segment intersections to determine where Jacobi fiber surfaces intersect~\cite{chattopadhyayAlgorithmCorrectComputation2024}.
It also requires ordering the intersections in the domain, which corresponds to the order of their respective segment intersections in the range due to the barycentric interpolation of the multivariate \pl map.

The robustness of Reeb space algorithms has not been fully addressed so far, even though it is a crucial step in deploying these algorithms in practice, in which accuracy must be guaranteed.
Robustness with respect to numerical errors is addressed~\cite{hristovArrangeTraverseAlgorithm2025} by computing arrangements with exact geometric computation using CGAL.
However, all four Reeb space algorithms assume a generic input map with no degenerate cases.
Two of them~\cite{hristovArrangeTraverseAlgorithm2025, tiernyJacobiFiberSurfaces2017} apply random numerical perturbations with finite precision arithmetic, which is not guaranteed to resolve all degenerate cases.
Most Reeb space algorithms mention that SoS~\cite{edelsbrunnerSimulationSimplicityTechnique1990a} can be used to ensure that the input map is generic, without providing explicit details. In this paper, we address this gap.

\paragraph*{Generic Bivariate \pl Maps.}
Generalizing the notion of genericity beyond the scalar case is a challenge. 
Although injectivity on the vertex set is a sufficient condition for a scalar \pl function to be generic, it does not eliminate all degenerate cases~\cite{tiernyJacobiFiberSurfaces2017} in multifields.
Furthermore, there is no known \pl analogue of stable maps from fiber topology~\cite{saekiTopologySingularFibers2004b}.
As a result, various definitions of generic PL maps have emerged in the literature, each tailored to the requirements of a specific algorithm.

The first definition of a generic \pl map~\cite{edelsbrunnerReebSpacesPiecewise2008a} is such that \emph{``the images of the vertices have no structural properties that can be removed by arbitrarily small perturbations ...''}.
Although the term \emph{structural properties} is not explicitly defined in that paper, its meaning is understood from the context. 
By considering the Reeb space as representing the topological structural properties of a \pl map, follow-up work formalizes this definition by requiring that no three points are collinear and no three segments are concurrent at their interiors~\cite{hristovHypersweepsConvectiveClouds2022}.
Both definitions imply that $f$ is injective on the edges and triangles of the input mesh because degenerate segments and triangles (see \cref{fig:segment-degenerate}) in the range introduce collinear points. 


This injectivity condition has also been used independently as a definition of genericity for bivariate \pl maps~\cite{chattopadhyayAlgorithmCorrectComputation2024}.
Another injectivity condition used in defining generic bivariate \pl maps requires that the images of no two edges be collinear~\cite{tiernyJacobiFiberSurfaces2017}. 
While these definitions~\cite{chattopadhyayAlgorithmCorrectComputation2024, tiernyJacobiFiberSurfaces2017} are sufficient for those Reeb space algorithms, they do not strictly guarantee that the Reeb space of a generic map remains the same under arbitrarily small perturbations~\cite{hristovHypersweepsConvectiveClouds2022}.

\section{Algebraic Formulation of Required Predicates}
\label{sec:required-predicates}

As discussed in \cref{sec:background}, the predicates required for computing bivariate Jacobi sets and Reeb spaces are lexicographic ordering, identifying upper and lower links of edges, and determining the order of segment intersections.
In this section, we will discuss which algebraic formulation we will use for each predicate.
\revision{Throughout the rest of the paper,} we assume that the input mesh $M$ is a triangulation of a $3$-manifold and that the $\pl$ map $f = (f_1, f_2): |M| \to \real^2$ is bivariate and not necessarily generic.

Unlike previous robust implementations of Jacobi set computation that use the definition of $h_{\lambda}$ to identify upper and lower links of edges, we will use the alternative formulation based on orientation tests (\cref{eqn:point_orientation}).
This choice is motivated by the fact that orientation tests are more widely supported in geometric software libraries, and their robust perturbation techniques are well-studied across all dimensions, including homogeneous coordinates.

The predicate for ordering segment intersections can also be handled using orientation tests. However, this approach requires explicitly computing the coordinates of the intersection points. 
As discussed in \cref{sec:comp-geom}, this is significantly less efficient than parameterizing the segments and directly comparing the values of the scalar parameters.
However, the resulting expression is not a determinant, which is a requirement for most symbolic perturbation schemes~\cite{edelsbrunnerSimulationSimplicityTechnique1990a, emirisEfficientPerturbationsHandling1997, alliezRemovingDegeneraciesPerturbing2000}.
Although this issue can be addressed, as demonstrated in \cref{sec:sos}, we introduce an alternative formulation of this predicate based on orientation tests in the dual projective plane.
Having both formulations enables a comparative evaluation of their correctness and practical efficiency.

\subsection{Order of Segments in Homogeneous Coordinates}
\label{sec:predicates-reeb-space}

An alternative method to determine the intersection order of three segments $s_{1, 2}, s_{3, 4}$, and $s_{5, 6}$ is to use homogeneous coordinates in the dual projective plane.
Given two points in homogeneous coordinates $[p_{1, 1}, p_{1, 2}, 1]$ and $[p_{2, 1}, p_{2, 2}, 1]$, the line passing through them is represented as $\ell_{1, 2} = [p_{1, 2} - p_{2, 2}, p_{2, 1} - p_{1, 1}, p_{1, 1}p_{2, 2} - p_{1, 2}p_{2, 1}]$.
Using the point-line duality in \revision{projective geometry~\cite[Ch.~8.2]{debergComputationalGeometryAlgorithms2008a}}, this line corresponds to a projective point $\ell_{1, 2}^*$ with the same coordinates in the dual plane.
Since this transformation preserves incidence and order, the orientation of the three points in homogeneous coordinates $\ell_{1, 2}^*, \ell_{3, 4}^*$ and $\ell_{5, 6}^*$ in the dual plane encodes the order of the original three segments $s_{1, 2}, s_{3, 4}$ and $s_{5, 6}$.


An example in \cref{fig:segment-order-projective} illustrates this construction.
Let $s_{1, 2}$ be a segment on the line $l_{1, 2}$ in the primal plane, whose dual point in the dual plane is $l^*_{1, 2}$. 
Each point on the segment corresponds to a line in the dual plane that passes through the point $l^*_{1, 2}$.
These lines sweep out two wedges, shown in gray in \cref{fig:sub2-homog} (\revision{see also \cref{fig:lines-wedges}}).
The three intersection points $i_1, i_2$, and $i_3$ on $s_{1, 2}$ map to three of these dual lines denoted as $i^*_1, i^*_2$, and $i^*_3$.
Meanwhile, the lines $\ell_{3, 4}$ and $\ell_{5, 6}$ (containing the segments $s_{3, 4}$ and $s_{5, 6}$) are mapped to the dual points $\ell^*_{3, 4}$ and $\ell^*_{5, 6}$, which lie on the lines $i^*_1$ and $i^*_2$, respectively.
The orientation of the dual points $l_{1,2}^*$ and $l_{3,4}^*$ matches the order of intersection points $i_1$ and $i_2$ along $s_{1,2}$ in the primal plane, which can be evaluated via the sign of the determinant:
\begin{equation}
\triangle_s = 
\det
\begin{pmatrix}
    p_{2,2} - p_{1,2} & p_{1,1} - p_{2,1} & p_{2,1} p_{1,2} - p_{1,1} p_{2,2} \\
    p_{4,2} - p_{3,2} & p_{3,1} - p_{4,1} & p_{4,1} p_{3,2} - p_{3,1} p_{4,2} \\
    p_{6,2} - p_{5,2} & p_{5,1} - p_{6,1} & p_{6,1} p_{5,2} - p_{5,1} p_{6,2}
\end{pmatrix}
\label{eqn:segment_order}
\end{equation}

\begin{figure}[htbp]
    \centering
    \begin{subfigure}[b]{0.49\columnwidth}
        \includegraphics[width=\columnwidth]{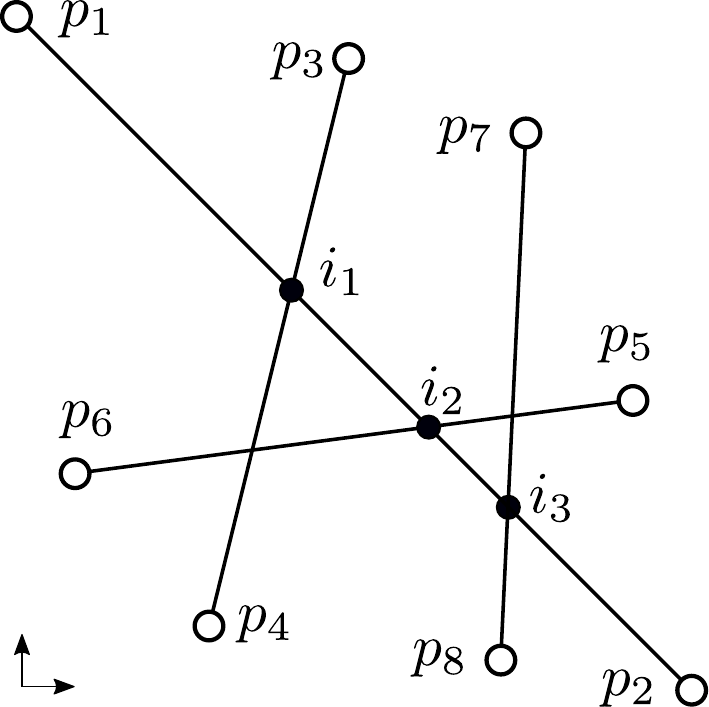}
        \caption{Segments in the plane.}
        \label{fig:sub1}
    \end{subfigure}
    \begin{subfigure}[b]{0.49\columnwidth}
        \includegraphics[width=\columnwidth]{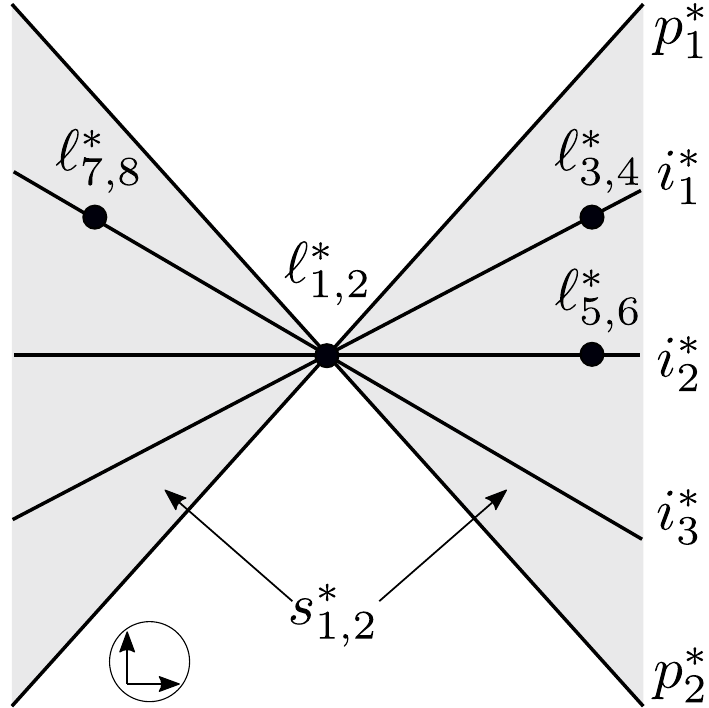}
        \caption{Dual Homogeneous Coordinates.}
        \label{fig:sub2-homog}
    \end{subfigure}
    \caption{
        The order of intersections of four segments $s_{1,2}$, $s_{3,4}$, $s_{5,6}$, and $s_{7,8}$ relative to the first segment $s_{1,2}$ can be computed with an orientation test in the dual projective plane.
        The segment $s_{1,2}$ maps to two wedges that meet at a point, as shown shaded in (b).
        We must flip the sign of the result if the points are not in the same wedge.
    }
    \label{fig:segment-order-projective}
\end{figure}

However, there is no guarantee that the lines corresponding to two segments will map to the same wedge in the dual plane.
For example, the lines of the segments $s_{5,6}$ and $s_{7,8}$ map to two points $l^*_{5, 6}$ and $l^*_{7, 8}$ that lie in different wedges in \cref{fig:segment-order-projective}.
In such cases, the orientation of the dual points is reversed relative to the order of the intersection points in the primal plane.
This can be resolved by two additional sign checks to determine whether the lines are mapped to the same wedge. This can be achieved by substituting their coordinates in the equation of the line given by $p^*_1$.
Geometrically, the wedge to which a line’s dual point maps to is determined by the line's orientation with respect to the line connecting the corresponding intersection point (in the primal plane) to the origin; see \cref{fig:lines-wedges}. 
We refer the reader to \cref{app:duality} for an additional discussion on special cases when one or more of the line segments pass through the origin in the primal plane result in ideal points at infinity in the dual plane.

\section{Symbolic Perturbation of Required Predicates}
\label{sec:schemes}

Having established the algebraic formulations of the predicates for computing Jacobi sets and Reeb spaces in \cref{sec:required-predicates}, we now demonstrate how to perturb a possibly degenerate \pl map $f: |M| \to \real^2$ to a generic \pl map $f_{\ep}$.
To ensure that $f_{\ep}$ is well-defined, we will only perturb the points (images of vertices) of $f$ (see \cref{sec:symbolic-perturb}).
\revision{
We show that this resolves all degenerate cases for points and segments (images of edges), shown in \cref{fig:segment-degenerate}.
We investigate perturbations using Edelsbrunner and M\"ucke's SoS (E)~\cite{edelsbrunnerSimulationSimplicityTechnique1990a} in \cref{sec:sos}, Alliez et al.'s perturbation of the world (A)~\cite{alliezRemovingDegeneraciesPerturbing2000} in \cref{sec:allienz}, and Yap's general symbolic perturbation scheme (YL, YT)~\cite{yapSymbolicTreatmentGeometric1990} in \cref{sec:yap}.
}
Since floating-point arithmetic cannot guarantee robustness for our predicates, we assume the use of exact arithmetic (see \cref{sec:comp-geom}).

\subsection{Simulation of Simplicity (E)}
\label{sec:sos}
In the SoS framework, each point $p_i\in P$ is perturbed via an $\ep$-expansion, defined as $p_i(\ep) = (p_{i, 1} + \ep_{i, 1}, p_{i, 2} + \ep_{i, 2})$ (see \cref{sec:comp-geom}).
This ensures that we can perform lexicographic point comparisons solely based on the first perturbed coordinates because they are unique.
To compare two points $p_i$ and $p_j$, where $i \neq j$, we first compare their unperturbed first coordinates $p_{i, 1}$ and $p_{j, 1}$. If they are equal, we break the ties using their indices $i$ and $j$, analogously to the comparison test widely used in the scalar case.
Orientation tests of three perturbed points are done by evaluating a perturbed determinant, which is discussed in the original SoS paper (see~\cite[Table-ii in the Appendix]{edelsbrunnerSimulationSimplicityTechnique1990a}).
\revision{
We show the number of terms and arithmetic operations in each row of the evaluation table of this scheme in \cref{table:evalution-table-point-orientation} (column E).
}
Next, we discuss the implementation of the segment order predicate through parametrization and dualization.

\subsubsection*{Parametrization}
The expression we obtain from the order of segments predicate with parameterization is not a determinant, unlike the typical setup in SoS.
However, this is not a problem if we can verify two properties that are guaranteed when we expand a determinant into an expression with SoS.
The first property is that at least one $\ep$-product has a constant coefficient.
This ensures that the predicate never evaluates to zero.
The second property is that no $\ep$-variable is raised to a power other than zero or one.
This ensures that they have a well-defined total ordering.
To verify these properties, we can explicitly expand the expression, but first we need to rearrange it to eliminate the denominator.

To expand the expression, we substitute all vectors with their perturbed counterparts, as described in \cref{sec:comp-geom}.
The resulting expression is $\frac{\mathbf{\qv(\ep) \times \sv(\ep)}}{\rv(\ep) \times \sv(\ep)} - \frac{\mathbf{\qpv(\ep) \times \spv(\ep)}}{\rv(\ep) \times \spv(\ep)}$.
To eliminate the denominator, we multiply both terms by their common denominator $(\rv(\ep) \times \spv(\ep))(\rv(\ep) \times \sv(\ep))$.
Although this changes the value of the expression, we are only concerned with its sign, which remains the same as long as the denominator is positive. Therefore, we must track the sign of both the denominator and the resulting expression:
$$(\qv(\ep) \times \sv(\ep))(\rv(\ep) \times \spv(\ep)) - (\qpv(\ep) \times \spv(\ep))(\rv(\ep) \times \sv(\ep)).$$

When expanded, the final expression has $768$ terms, making the evaluation of its sign non-trivial.
For this reason, we use an automatic evaluation procedure (\cref{alg:polynomial-sign-evaluation}) whose description we postpone to \cref{sec:polynomial-sign} for clarity of exposition.
The evaluation table of the final expression has a maximum depth of $34$.
We compare the number of terms and arithmetic operations (addition and multiplication) of each row of the evaluation tables of all perturbation schemes in \cref{table:evalution-table-segment-order} (column E).
Finally, we confirm that no $\ep$-variable in the expression has an exponent different than zero or one, and that there is an $\ep$-product that has a constant coefficient.

\subsubsection*{Dualization}
The second approach we introduced in \cref{sec:required-predicates} for implementing the segment order predicate relies on orientation tests in the dual projective plane. This method fits naturally within the SoS framework, as it can be formulated as a determinant. However, unlike standard orientation predicates, we cannot use the evaluation tables provided by Edelsbrunner and Mücke~\cite{edelsbrunnerSimulationSimplicityTechnique1990a}, since the determinant in this case does not directly involve the original input points.
We can still use \cref{alg:polynomial-sign-evaluation} to automatically determine the sign of the determinant of the matrix described in \cref{eqn:segment_order} using the perturbed coordinates. 

Similar to the parametrized formulation, we need to compute additional signs to decide whether to flip the sign of the final result.
This depends on whether the lines of the two segments are mapped to the same wedge in the dual plane (\cref{sec:predicates-reeb-space}).
We also need to check for the signs that account for normalizing the points in homogeneous coordinates~\cite{edelsbrunnerSimulationSimplicityTechnique1990a}.

The fully expanded expression for the dualized formulation of the segment order predicate contains 768 terms, and its evaluation table reaches a maximum depth of 34 as shown in \cref{table:evalution-table-segment-order}. These are exactly the same figures as those obtained from the parametrization-based formulation of the predicate.
When comparing the two expanded expressions and their corresponding evaluation tables, we observe that they are identical.

\begin{table}[ht]
\centering
\scalebox{0.8}{
\begin{tabular}{crrrrrrrr}
\toprule
            & \multicolumn{4}{l}{\textbf{Number of Terms}}                                                                               & \multicolumn{4}{l}{\textbf{Number of Operations}}                                                               \\ 
\midrule
\textbf{Depth} & \textbf{A} & \textbf{E} & \textbf{YL} & \textbf{YT} & \textbf{A} & \textbf{E} & \textbf{YL} & \textbf{YT} \\ 
\midrule
0           & 6         & 6         & 6          & 6          & 11       & 11        & 11         & 11         \\ 
1           & 18        & 2         & 2          & 2          & 47       & 1         & 1          & 1          \\ 
2           & 12        & 2         & 2          & 2          & 35       & 1         & 1          & 1          \\ 
\midrule
3           &         & 2         & 2          & 2          &        & 1         & 1          & 1          \\ 
4           &         & 1         & 1          & 2          &        & 0         & 0          & 1          \\ 
\midrule
5           &         &          &           & 2          &        &          &           & 1          \\ 
6           &         &          &           & 2          &        &          &           & 1          \\ 
7           &         &          &           & 1          &        &          &           & 0          \\
\midrule
Non-zero constant           &            &    1      &  -1         & -1          &            &          &           &           \\
\bottomrule
\end{tabular}
}
\caption{
    Number of terms and arithmetic operations for each row in the evaluation tables of each symbolic perturbation scheme for the point orientation predicate.
    The perturbation schemes are labeled according to the initial of their first author: A~\cite{kettnerClassroomExamplesRobustness2008}, E~\cite{edelsbrunnerSimulationSimplicityTechnique1990a}, YL~\cite{yapSymbolicTreatmentGeometric1990} with lexicographic ordering and YT~\cite{yapSymbolicTreatmentGeometric1990} with total ordering.
}
\label{table:evalution-table-point-orientation}
\end{table}

\subsection{Perturbing the World (A)}
\label{sec:allienz}

Next, we discuss the symbolic perturbation scheme called perturbing the world, which is specifically designed for line segments~\cite{alliezRemovingDegeneraciesPerturbing2000}.
This scheme uses a single global epsilon $\ep$ and the epsilon expansion is given as $p_i(\ep) = \big(p_{i, 1} + \ep p_{i, 2}, p_{i, 2} + \ep^2 p_{i, 1} + \ep^3(p_{i, 1}^2 + p_{i, 2}^2)\big)$.
It is important to note that this scheme assumes that the input contains no duplicate points.
Indeed, from the structure of the epsilon expansion, a lexicographic comparison would be unable to distinguish between two points that share the same coordinates.
The paper~\cite{alliezRemovingDegeneraciesPerturbing2000} gives details about evaluating the point orientation predicate, and we can use \cref{alg:polynomial-sign-evaluation} to automatically generate the evaluation table for the segment order predicate.
As we have seen previously, both formulations of that predicate produce the same final expression.
When we generate the evaluation table from the final expression, the ordering of all epsilon products is $\ep, \ep^2, \dots, \ep^n$, where $n$ is the maximum degree of any epsilon in the expression.
This scheme produces significantly more terms than the other two schemes, as can be seen in column A of \cref{table:evalution-table-point-orientation} and \cref{table:evalution-table-segment-order}.

\begin{table}[!ht]
\centering
\scalebox{0.8}{
\begin{tabular}{crrrrrrrr}
\toprule
            & \multicolumn{4}{l}{\textbf{Number of Terms}}                                 & \multicolumn{4}{l}{\textbf{Number of Operations}}                \\ 
\midrule
\textbf{Depth} & \textbf{A} & \textbf{E} & \textbf{YL} & \textbf{YT} & \textbf{A} & \textbf{E} & \textbf{YL} & \textbf{YT} \\  
\midrule
0   & 48  & 48  & 48 & 48 & 191  & 191 & 191 & 191 \\
1   & 240 & 16  & 16 & 16 & 1199 & 47  & 47  & 47  \\
2   & 192 & 16  & 16 & 16 & 959 & 47  & 47  & 47  \\
3   & 432 & 16  & 16 & 16 & 2303 & 47  & 47  & 47  \\
4   & 576 & 4   & 8  & 16 & 3263 & 7   & 15  & 47  \\
5   & 192 & 6   & 16 & 16 & 1151 & 11  & 47  & 47  \\
\midrule
6   &     & 16  & 8  & 16 &      & 47  & 15  & 47  \\ 
$\vdots$ & & $\vdots$ & $\vdots$ & $\vdots$ & & $\vdots$ & $\vdots$ & $\vdots$ \\ 
34  &     & 1   & 6  & 6  &      & 0   & 11  & 11  \\
\midrule
35  &     &     & 4  & 4  &      &     & 7   & 7   \\ 
$\vdots$ & & & $\vdots$ & $\vdots$ & & & $\vdots$ & $\vdots$ \\ 
59  &     &     & 1  & 4  &      &     & 0   & 7   \\ 
\midrule
60  &     &     &    & 6 &      &     &     & 11  \\ 
$\vdots$ & & &  & $\vdots$ & & &  & $\vdots$ \\ 
163 &     &     &    & 1  &      &     &     & 0   \\ 
\midrule
Non-zero constant & & 1 & -1 & -1 & & & & \\
\bottomrule
\end{tabular}
}
\caption{
    Number of terms and arithmetic operations for each row in the evaluation tables of each symbolic perturbation scheme for the segment order predicate.
    The perturbation schemes are labeled according to the initial of their first author: A~\cite{kettnerClassroomExamplesRobustness2008}, E~\cite{edelsbrunnerSimulationSimplicityTechnique1990a}, YL~\cite{yapSymbolicTreatmentGeometric1990} with lexicographic ordering, and YT~\cite{yapSymbolicTreatmentGeometric1990} with total ordering.
}
\label{table:evalution-table-segment-order}
\end{table}

\subsection{General Symbolic Scheme (YL, YT)}
\label{sec:yap}

The general symbolic scheme proposed by Yap~\cite{yapGeometricConsistencyTheorem1990, yapSymbolicTreatmentGeometric1990} does not require an epsilon expansion and relies on symbolic differentiation, which is present in most symbolic mathematics libraries.
The only prerequisite to using this scheme is to choose an admissible ordering of the power products of partial derivatives.
Two common admissible orderings are the lexicographic ordering (YL) and the total degree ordering (YT), where power products are first sorted by their total degree and then lexicographically.
When applied to the lexicographic order predicate, this scheme produces the same result as SoS --- ties in the first coordinate are broken by its index.

We show the difference in the number of terms and arithmetic operations produced by the two admissible orderings for the orientation of points predicate in \cref{table:evalution-table-point-orientation} and for the order of segments predicate in \cref{table:evalution-table-segment-order}.
In the evaluation table for the lexicographic ordering (YL), rows with fewer terms appear earlier because higher-order derivatives are evaluated before lower-order partial derivatives.
As a result, simpler expressions tend to appear earlier in the table, which is advantageous for predicates that can be resolved without requiring deep evaluations.
The lexicographic ordering approach also has a smaller maximum depth because partial derivatives of higher order are evaluated first.


\section{Genericity}
\label{sec:genericity}
The results from the previous section allow us to perturb any bivariate \pl map $f$ to a generic \pl map $f_{\ep}$.
We will use the definition of genericity that requires that no three points are collinear and no three segments are concurrent in their interiors~
\cite{hristovHypersweepsConvectiveClouds2022}.
As discussed in \cref{sec:multifield-topo}, this implies all other formulations of genericity of bivariate \pl maps~\cite{edelsbrunnerReebSpacesPiecewise2008a, chattopadhyayAlgorithmCorrectComputation2024, tiernyJacobiFiberSurfaces2017}.

\begin{lemma}
Let $M$ be a triangulation of a $3$-manifold and let $f : |M| \to \real^2$ be a bivariate \pl map.
The symbolic perturbation $f_{\ep}$ of $f$ with either Simulation of Simplicity or general symbolic perturbation is a generic \pl map.
\end{lemma}
\begin{proof}
    As we have shown in \cref{sec:schemes} the lexicographic order, orientation of points, and order of segments predicates never evaluate to zero.
    This is evident from \cref{table:evalution-table-point-orientation} and \cref{table:evalution-table-segment-order} where the predicates reach a non-zero constant term in their evaluation tables.
    Therefore, no three points are collinear and no three segments are concurrent in their interiors. So, $f_{\ep}$ is generic.
\end{proof}

\section{Sign Evaluation of Symbolic Polynomials}
\label{sec:polynomial-sign}





Evaluating the perturbed expressions of the parametrized and dualized formulations of the order of segments predicate is challenging because they involve points constructed from the original input.
As we have shown in \cref{sec:schemes} this leads to symbolic expressions with hundreds of terms, which are impractical to factor manually since it is not straightforward to express them as sums of subdeterminants.
In this section, we demonstrate how this factorization can be automated for the SoS and perturbing the world schemes.
We provide a concrete example at the end of the section.
Note that the general symbolic scheme does not use an explicit symbolic perturbation and can be resolved by directly differentiating the expression of the predicate in terms of the original coordinates.

Consider an algebraic expression $q$ used to evaluate a particular geometric predicate.
Typically, $q$ is the expansion of a determinant, but here we consider the general case where it is a multivariate polynomial in the original coordinates.
Let $q_{\ep}$ be the multivariate polynomial obtained by evaluating $q$ on the $\ep$-expansion of the input coordinates.
We assume that a total ordering is defined on all $\ep$-products in $q_{\ep}$.
Instead of manually factoring $q_{\ep}$, we use computer algebra systems such as Mathematica~\cite{mathematica} or SymPy~\cite{meurer2017sympy} to automate the steps for producing the relevant evaluation tables.


For each $\ep$-product $w$ in the ordered sequence of all $\ep$-products we define $T_w$ to be the sum of all terms in $q_{\ep}$ which contain exactly $w$ and no other $\ep$-variables.
Then $T_w = c \times w$, where $c$ is an expression that only involves the original point coordinates.
Note that the first $\ep$-product $w=1$ contains no $\ep$-variables and corresponds to evaluating the predicate on the non-perturbed coordinates.
If $c$ is non-zero, we return its sign.
Otherwise, we proceed down the list until a non-zero $c$ is found.

A non-zero coefficient $c$ is guaranteed to be found if there is at least one $\ep$-product in $q_{\ep}$ that is multiplied by a constant.
Once such a term is found, there is no need to evaluate further terms, as they are guaranteed to be orders of magnitude smaller.
Therefore, our procedure can be used to evaluate predicates more general than those expressible as determinants~\cite{edelsbrunnerSimulationSimplicityTechnique1990a}.
For determinant predicates, we can compute the table of relevant determinants produced by the SoS method by listing all coefficients $c$ in order.
We will refer to the table produced by our method more generally as an \emph{evaluation table} 
since it does not need to include any subdeterminant.
We illustrate this procedure in \cref{alg:polynomial-sign-evaluation}.

\begin{algorithm}
\caption{Evaluation Table of a Symbolic Polynomial}
\label{alg:polynomial-sign-evaluation}
\begin{algorithmic}[]
    \Procedure{ComputeEvaluationTable}{$q_{\ep}$}
    \State $W \gets $ ordered list of all $\ep$-products in $q_{\ep}$.
    \State $E \gets [\ ]$ \Comment{Evaluation Table}
    \For{$w \in W$}
        \State $T_w \gets \textnormal{all terms of $q_{\ep}$ that contain exactly }w$
        \State $c \gets T_w / w$ 
        \State $E$.add($c$)
        \If{$c \textnormal{ is a constant }$}
            \State \textbf{return} $E$
        \EndIf
    \EndFor
    \EndProcedure
\end{algorithmic}
\end{algorithm}

Finally, we present an example.
Let $p_i(\ep) = (p_{i, 1} + \ep_{i, 1}, p_{i, 2} + \ep_{i, 2})$ and $p_j(\ep) = (p_{j, 1} + \ep_{j, 1}, p_{j, 2} + \ep_{j, 2})$ be the $\ep$-expansions of two points with SoS.
Suppose that we have an expression for a predicate $q_{\ep}$ where the terms are ordered and grouped by their $\ep$-product ordering.
Let $q_{\ep}$ = ($p_{i, 1} + p_{j, 2}) + (p_{i, 1}\ep_{i, 1} + p_{i, 2}\ep_{i, 1}) + \revision{10\ep_{i, 1}\ep_{j, 2}}  + p_{i, 1}\ep_{j, 2}$.
First, we evaluate $p_{i, 1} + p_{j, 2}$.
If that is zero, we then evaluate $p_{i,1} + p_{i,2}$, and return its sign if it is nonzero.
Finally, $\ep_{i, 1}\ep_{j, 2}$ has a constant coefficient of $10$, so we return a positive sign.
We do not evaluate the coefficient of $\ep_{j,2}$, since $\ep_{i,1}\ep_{j,2} \gg \ep_{j,2}$.

\section{Implementation}
\label{sec:implementation}


We implemented all predicates and perturbation schemes discussed \cref{sec:schemes}, as well as \cref{alg:polynomial-sign-evaluation}, using the SymPy symbolic mathematics package~\cite{meurer2017sympy} in Python.
Symbolic variables were used to represent the $\ep$-expansion of all coordinates, utilizing SymPy's built-in capabilities for factoring, simplifying, and differentiating symbolic expressions.
We used Python’s \texttt{itertools} library to generate orderings of $\ep$-products and partial derivatives.
The output of our code is a set of evaluation tables for all required geometric predicates.
\revision{
Our code is freely available on \href{https://github.com/peter-hristov/symbolic-perturbation-for-bivariate-computational-topology}{GitHub}~\cite{hristovSymbolicPerturbation2025}.
}

The evaluation tables for each predicate can be used to compute perturbed signs by substituting numerical values for the symbolic coordinates.
However, doing this directly in Python is slow.
In order to generate high-performance C++ code, we advise the use of SymPy's built-in \texttt{ccode} functionality.
The rows of each evaluation table can be organized into numerous cases in a \texttt{switch} statement and exported as C++ functions in a header file.
This header file can then be integrated with existing code for computing Jacobi sets of the Reeb spaces in order to handle the predicates robustly.


\begin{table*}[h]
    \centering
    \scalebox{0.8}{
    \begin{tabular}{lcrrrrrr}
        \toprule
        \textbf{Dataset} & 
        \textbf{Dimensions} & 
        \textbf{Tetrahedra} & 
        \textbf{Vertices} & 
        \textbf{Edges} & 
        \textbf{Duplicates} & 
        \textbf{Collinear Points} & 
        \textbf{Concurrent Segments} \\
        \midrule
        engine       & 255×255×109 & 35,438,625 & 7,208,960  & 42,888,814  & 5,884,231   & 1,050,940,003& 57    \\
        enzo         & 256×256×256 & 82,906,875 & 16,777,216 & 100,074,240 & 86,344      & 9,832,907    & 0     \\
        ethane-diol  & 115×116×134 & 8,718,150  & 1,787,560  & 10,592,843  & 92          & 1            & 0     \\
        isabel1      & 125×125×100 & 7,611,120  & 1,562,500  & 9,253,475   & 6,631       & 537,992      & 0     \\
        scission     & 40×40×66    & 494,325    & 105,600    & 613,106     & 104,665     & 412,777,680  & 16,280\\
        tooth        & 103×94×161  & 7,588,800  & 1,558,802  & 9,228,973   & 5,255       & 35,780,055   & 1     \\
        \bottomrule
    \end{tabular}
    }
    \caption{
        Number of degenerate cases found for each dataset using random sampling with $n = 10^{10}$.
        Duplicates refer to the number of vertices that map to the same point in the range as another vertex.
        Collinear refers to the number of collinear triplets of points in the range for $n$ samples of triplets.
        Collinear refers to the number of collinear triplets of segments in the range for $n$ samples of triplets. 
    }
    \label{tab:geom-full}
\end{table*}


\section{Evaluation}
\label{sec:evaluation}

So far, we have established three important things about the symbolic perturbation of our predicates.
The first is that both the parametrized and dualized formulations of the order of segments predicate have the same expression.
The second is that the general symbolic scheme could be more efficient in practice when implemented with the lexicographic ordering, rather than with the total ordering, because the earlier rows in the table have fewer terms.
The third is that the perturbation scheme that perturbs the world is less practical because it requires significantly more arithmetic operations than the other two schemes.

The goal of this section is to evaluate the practical performance of our proposed predicates with two experiments.
First, we would like to understand how often degenerate cases occur in real-world datasets.
This gives us an idea of how often perturbation is required to evaluate them.
Our second goal is to determine how many terms and arithmetic operations are needed on average to resolve degenerate cases.
In order to estimate this, we run our predicates on synthetically generated degenerate cases.

\subsubsection*{Real-world Data}
We have chosen multivariate datasets widely used in the scientific visualization community~\cite{carrFiberSurfacesGeneralizing2015, jankowaiFeatureLevelSetsGeneralizing2019}.
All datasets are sampled on the vertices of a regular cubic grid, which we then subdivide using five tetrahedra per cube in ParaView~\cite{Ahrens2005ParaView}.
In the first five columns of \cref{tab:geom-full} we list the original grid resolution and the number of simplices in the resulting tetrahedral meshes.
In the engine, scission, and tooth datasets, one of the scalar fields is the gradient magnitude of the primary scalar field. 

To determine the number of degenerate cases for each dataset, we would need to count the number of duplicate points, three or more collinear points, and three or more segments that are concurrent in their interiors (see \cref{fig:segment-degenerate}).
Aside from counting duplicate points, which can be done efficiently with a hash map, it is infeasible to count the exact number of degenerate cases for the other two predicates.
This is because the number of triplets for up to 16 million vertices and 100 million edges is on the order of $10^{17} \textnormal{ to } 10^{19}$.
Generating all such triplets is impractical and using memorization to test pairs instead of triplets requires terabytes of memory. 

In our experiment, we estimate the number of degenerate cases by generating $10^{10}$ random triplets of vertices and edges.
We then used CGAL~\cite{cgal-manual} to test for collinearity and concurrency in the corresponding points and segments using an exact kernel.
As shown in \cref{fig:segment-order}, the number of duplicate and collinear points is significant in some datasets, particularly those where one scalar field is the gradient of the other.
In contrast, there are far fewer degenerate cases for concurrent segments. 
This can be explained by the fact that, in such cases, six points (the segment endpoints) must align, rather than just two or three.
Although we report zero degenerate cases involving concurrent segments in three of the datasets, we cannot guarantee that none exist.
We may simply have not sampled them.

We can confirm that degenerate cases of all types (see \cref{fig:segment-degenerate}) do occur in real-world bivariate datasets.
\revision{
In particular, highly correlated scalar fields are more likely to have similarities that could lead to degenerate cases.
}
Although degenerate configurations are, in a sense, local, they can affect the global combinatorial structure of Jacobi sets and Reeb spaces.
This clearly demonstrates the need for robust predicates introduced in this paper.

\subsubsection*{Synthetic Data}

Having established that all degenerate cases occur in real-world datasets, we now evaluate how efficiently they can be resolved using the perturbation schemes proposed in \cref{sec:schemes}.
For this evaluation, we used synthetically generated degenerate cases.
For collinear points, we randomly generated three rational points along the $y=x$ line and then applied a random rotation and translation.
For concurrent segments, we used three random rational numbers $t_1, t_2$ and $t_3$ to generate three rational points on the unit circle with a stereographic projection, where $x_i = (1 - t_i^2)/(1 + t_i^2)$ and $y_i = (2t_i)/(1 + t_i^2)$ for $i \in \{1, 2, 3\}$.
Taking each point and its antipodal results in three concurrent segments, which we then randomly translated and rotated.
Using rational numbers in all cases is crucial to ensure robustness (see \cref{sec:comp-geom}).

We evaluated Yap's general scheme~\cite{yapSymbolicTreatmentGeometric1990} with both lexicographic (YL) and total ordering (YT), as well as the SoS scheme~\cite{edelsbrunnerSimulationSimplicityTechnique1990a}.
We did not evaluate the perturb the world scheme, because, as we saw in \cref{sec:schemes}, it produces significantly more terms.
The results of our experiments are shown in \cref{table:depth-operations} and \cref{table:depth-histogram}.
In \cref{table:depth-operations}, we report statistics on the depth and number of arithmetic operations across one million randomly generated degenerate cases, for both collinear points and concurrent segments.
Arithmetic operations were counted as the number of additions and multiplications present in each symbolic expression.
In practice, addition and multiplication may have different costs, depending on library-specific implementation details.
All computations were performed using SymPy's rational types to ensure exact arithmetic.

The performance of all perturbation schemes is almost identical, as we can see from \cref{table:depth-operations}.
The average number of operations and standard deviation differ by very small margins.
Overall, SoS uses the fewest arithmetic operations, closely followed by the general symbolic scheme with total and lexicographic orderings.
Interestingly, lexicographic ordering does not perform better than total ordering as predicted.
The smaller maximum depth of the lexicographic ordering is never reached as we can see from \cref{table:depth-histogram}. 
Furthermore, larger expressions earlier in the table, which we see with total ordering, may actually be beneficial, since they can help resolve a case with less depth.
Finally, we demonstrate in \cref{table:depth-histogram} that $99.9\%$ of all degenerate cases are resolved with only one additional evaluation (depth one).
This demonstrates that the schemes we have proposed are efficient and do not introduce a substantial overhead.

\begin{table}[ht]
    \centering
    \scalebox{0.8}{
    \begin{tabular}{ccrrrr}
        \toprule
        \multicolumn{2}{c}{} & 
        \multicolumn{2}{c}{\textbf{Depth}} & 
        \multicolumn{2}{c}{\textbf{\#Operations}} \\
        \cmidrule(lr){3-4} \cmidrule(lr){5-6}
        \textbf{Predicate} & \textbf{Scheme} & 
        \textbf{Mean} & \textbf{SD} & 
        \textbf{Mean} & \textbf{SD} \\
        \midrule
        \multirow{3}{*}{Orient} 
            & YL & 1.001 & 0.001 & 11.001 & 0.030 \\
            & YT & 1.001 & 0.001 & 11.001 & 0.030 \\
            & E  & 1.001 & 0.001 & 11.001 & 0.030 \\
        \midrule
        \multirow{3}{*}{Order} 
            & YL & 1.001 & 0.040 & 191.051 & 1.697 \\
            & YT & 1.001 & 0.035 & 191.051 & 1.646 \\
            & E  & 1.001 & 0.032 & 191.046 & 1.488 \\
        \bottomrule
    \end{tabular}
    }
    \caption{
        Mean and standard deviation (SD) of depth and number of arithmetic operations by predicate and scheme.
    }
    \label{table:depth-operations}
\end{table}

\begin{table}[ht]
    \centering
    \scalebox{0.8}{
    \begin{tabular}{ccrrrrrrr}
        \toprule
        \multicolumn{2}{c}{} & 
        \multicolumn{7}{c}{\textbf{Depth Count}} \\
        \cmidrule(lr){3-9}
        \textbf{Predicate} & \textbf{Scheme} & 
        \textbf{1} & \textbf{2} & \textbf{3} & \textbf{5} & \textbf{6} & \textbf{7} & \textbf{14} \\
        \midrule
        \multirow{3}{*}{Orient} 
            & YL & 999,087 & 913 & 0 & 0 & 0 & 0 & 0 \\
            & YT & 999,087 & 913 & 0 & 0 & 0 & 0 & 0 \\
            & E  & 999,087 & 913 & 0 & 0 & 0 & 0 & 0 \\
        \midrule
        \multirow{3}{*}{Order} 
            & YL & 998,961 & 1027  & 0  & 0  & 0 & 11 & 1 \\
            & YT & 998,961 & 1027  & 0  & 11 & 1 & 0  & 0 \\
            & E  & 999,030 & 959   & 11 & 0  & 0 & 0  & 0 \\
        \bottomrule
    \end{tabular}
    }
    \caption{        
        Depth counts by predicate and scheme.
    }
    \label{table:depth-histogram}
\end{table}

\section{\revision{Discussion and Limitations}}
\label{sec:discussion}

In our evaluation in \cref{sec:evaluation} we have confirmed that degenerate cases do occur in real-life data and that they can be handled robustly without substantial overhead.
However, symbolic perturbation has faced multiple concerns in the computational geometry literature~\cite{schornDegeneracyGeometricComputation1994, seidelNatureMeaningPerturbations1998}.
First, it complicates the debugging process during algorithm development.
For example, visualizing the perturbed arrangement of segments becomes increasingly difficult.
Second, we compute the topological structure of the perturbed data, not the original.
\revision{
This is an important distinction, even if we assume that real-world data contains noise, because, to the best of our knowledge, there are no stability results that allow us quantify the effects of a perturbation for Jacobi sets and Reeb spaces.}
Furthermore, perturbation introduces artificial features and there are currently no established simplification methods that can eliminate them.
Our approach also relies on exact arithmetic, and its overall impact on algorithm performance \revision{requires further investigation}.






Modern computational geometry libraries such as CGAL~\cite{cgal-manual} treat degenerate cases as legitimate parts of the input and handle them explicitly, without relying on perturbation.
This approach is generally faster and conceptually cleaner, but only when all possible degenerate cases are known and can be classified.
It remains to be seen whether classifying all degenerate cases for Reeb spaces and Jacobi sets is tractable, even in the bivariate case.

\revision{
Finally, note that our approach applies only to bivariate \pl maps.
For higher-dimensional \pl maps, supporting the geometric predicates required for higher-dimensional geometric arrangements will be necessary.
In other cases, such as cubic meshes with a trilinear interpolant, the set of Jacobi edges is not guaranteed to be a subset of the domain’s edges.
However, this is currently not an issue since all Jacobi set and Reeb space algorithms assume a \pl map.
}

\section{Conclusion}
\label{sec:conclusion}

In this paper, we presented both the theoretical foundation and practical implementation of \revision{robust predicates} for bivariate Jacobi sets and Reeb spaces algorithms.
We surveyed relevant work in computational geometry on robust geometric predicates and highlighted the unique challenges these predicates face in multivariate computational topology, particularly the need for exact geometric computation and the emergence of complex expressions under symbolic perturbation.
To handle these nontrivial expressions, we introduced a general procedure for evaluating them automatically, rather than manually, as is typically done in the literature.

We demonstrated that all predicates required for computing Jacobi sets and Reeb spaces can be handled robustly using various symbolic perturbation schemes, including, but not limited to, the well-known Simulation of Simplicity.
Then, we assessed the practical performance of these schemes and confirmed that degenerate cases do, in fact, arise in real-world bivariate datasets. However, we also demonstrated that most degenerate cases are resolved at very shallow depths requiring only a few additional operations. 
Finally, we outlined a procedure for generating efficient C++ code to evaluate these predicates without using any symbolic arithmetic.

\acknowledgments{
   This work was partially supported by the Wallenberg Autonomous Systems and Software Program (WASP) funded by the Knut and Alice Wallenberg Foundation; the Swedish e-Science Research Center (SeRC); and the Swedish Research Council (VR) grant 2023-04806.
}

\bibliographystyle{abbrv-doi}

\bibliography{template}

\appendix

\newpage

\section{Background on \pl Topology}
\label{app:background-pl}

\revision{
In this appendix we review some foundational definitions from \pl topology.
For a comprehensive introduction, we refer the reader to relevant textbooks~\cite{edelsbrunnerComputationalTopology2009, rourke2012introduction}.
}

\revision{
A \emph{simplex} $\sigma$ of dimension $\sd$ is a convex combination of $\sd$ affinely independent points in $\real^{\sd+1}$.
The points that define $\sigma$ are called its \emph{vertices}, and we say that a simplex is \textit{spanned} by its vertices.
A subset of a simplex spanned by a subset of its vertices is called a \emph{face} of that simplex.
For example, simplices of dimension $0, 1, 2$ and $3$ are points, line segments, triangles, and tetrahedra respectively.
}

\revision{
A simplicial complex $K$ is a collection of simplices where the faces of any simplex are also in $K$ and the intersection of any two simplices is either empty or a face of both.
The underlying geometric space of $K$ is denoted as $|K|$.
We say that $K$ triangulates a manifold when $|K|$ is homeomorphic to a manifold.
The notion of a topological neighbourhood is defined for a simplicial complex by using the \textit{star} and \textit{link} of a simplex $\sigma$.
The star of $\sigma$ is the set of simplices that have $\sigma$ as a face, and the link of $\sigma$ consists of the faces of all simplices in the star of $\sigma$ that do not intersect $\sigma$.
}

\revision{
A scalar function $f : |K| \to \real$ is \pl when its values are only defined on the vertices of $K$ and extend to the interiors of the simplices via barycentric interpolation.
\textit{Critical vertices} of scalar \pl functions are defined using the upper and lower link of a vertex.
The \textit{upper link} of a vertex $v$ consists of all simplices $\sigma$ in the link of $v$ whose vertices have a higher function value than $f(v)$.
The \textit{lower link} is defined analogously with all simplices $\sigma$ in the link of $v$ whose vertices have a lower function value than $f(v)$.
}

\section{Order of Segments with Ideal Points}
\label{app:duality}
Example of when the lines of segments are be mapped to different wedges (see \cref{sec:predicates-reeb-space}). 

\begin{figure}[htbp]
    \centering
    \begin{subfigure}[b]{0.49\columnwidth}
        \includegraphics[width=\columnwidth]{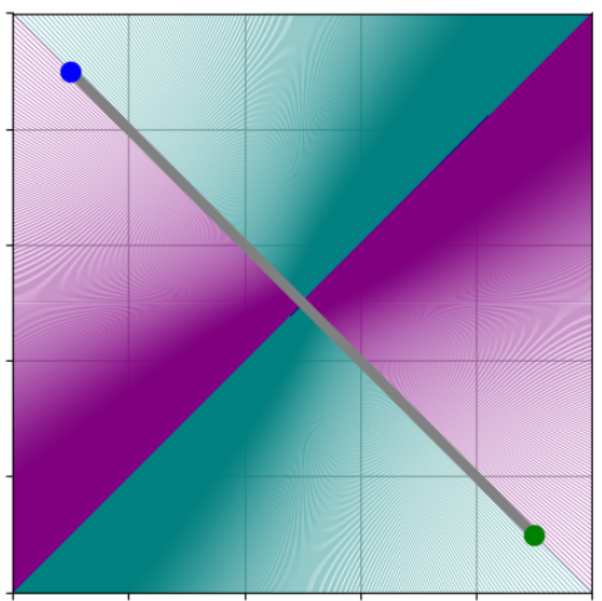}
        \caption{Lines in the plane.}
        \label{fig:sub1}
    \end{subfigure}
    \begin{subfigure}[b]{0.49\columnwidth}
        \includegraphics[width=\columnwidth]{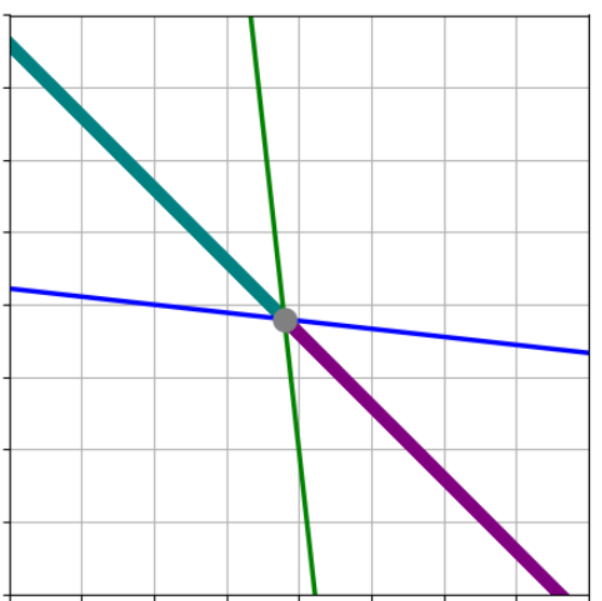}
        \caption{Dual Homogeneous Coordinates.}
        \label{fig:sub2}
    \end{subfigure}
    \caption{
        Given a segment with blue and green endpoints, we show which lines $\ell$ intersecting the halfway point of the segment are mapped in either of the two wedges that correspond to the segment in the dual projective plane.
        As $\ell$ becomes close to the line defined by the segment, its corresponding point in the dual projective plane becomes closer to the point that corresponds to the line defined by the segment in the primal plane.
        As $\ell$ becomes close to the line that connects the halfway point and the origin it moves away from the gray point in the dual projective plane and into infinity.
    }
    \label{fig:lines-wedges}
\end{figure}

\paragraph{Orientation Tests with Ideal Points.}
In the dual formulation of the segment order predicate described in \cref{sec:predicates-reeb-space}, lines that pass through the origin in the primal plane correspond to ideal points at infinity in the dual plane. This requires special handling, as orientation tests are not well defined in that case~\cite{edelsbrunnerSimulationSimplicityTechnique1990a}.
This is a consequence of the projective plane being a non-orientable topological space.
Nevertheless, the determinant used in orientation tests can still be meaningfully interpreted. 
When all three points are at infinity, the determinant will be zero, and this case can be perturbed with SoS.

When two points are at infinity, they can be geometrically interpreted as direction vectors in the affine plane at $z=0$.
An orientation between two direction vectors and a finite homogeneous point reduces to computing the determinant of the direction vectors scaled by the last coordinate of the finite point.
This is well defined both algebraically and geometrically. 
Geometrically, it is equivalent to the points $\ell^*_{3, 4}$ and $\ell^*_{5, 6}$ form \cref{fig:segment-order-projective} sliding along the lines $i^*_1$ and $i^*_2$ respectively to infinity (\revision{see \cref{fig:lines-wedges}}).
Algebraically, this follows from the expression:
\[
\det
\begin{pmatrix}
x_1 & y_1 & z_1 \\
x_2 & y_2 & 0 \\
x_3 & y_3 & 0
\end{pmatrix}
= 
z_1
\det
\begin{pmatrix}
    x_2 & y_2 \\
    x_3 & y_3
\end{pmatrix}
.
\]
We do not need to explicitly account for the sign of $z_1$ because it is already accounted for in the evaluation of the orientation test in homogeneous coordinates~\cite{edelsbrunnerSimulationSimplicityTechnique1990a}.
However, we need to normalize the direction vectors $[x_2, y_2, 0]$ and $[x_3, y_3, 0]$ using their first non-zero coordinate with two additional sign checks.

When one point is at infinity, the orientation test reduces to comparing the orientation of the affine vector defined by the two finite points and the affine vector defined by the point at infinity because:
\[
\det
\begin{pmatrix}
x_1 & y_1 & z_1 \\
x_2 & y_2 & z_2 \\
x_3 & y_3 & 0
\end{pmatrix}
= 
z_1 z_2
\det
\begin{pmatrix}
    \frac{x_2}{z_2} - \frac{x_1}{z_1} & \frac{y_2}{z_2} - \frac{y_1}{z_1} \\
    x_3 & y_3
\end{pmatrix}
.
\]
As in the previous case, we do not need to explicitly account for the signs of $z_1$ and $z_2$, because they will be accounted for in the evaluation of the determinant, but we do need to account for the sign of normalizing the direction vector (the point at infinity).

\end{document}